\documentclass[12pt]{amsart}
\usepackage{amsmath,amsthm,amssymb}

\usepackage[utf8]{inputenc} 
\usepackage[T1]{fontenc}    
\usepackage{hyperref}       
\usepackage{url}           
\usepackage{booktabs}       
\usepackage{amsfonts}       
\usepackage{nicefrac}       
\usepackage{microtype}      
\usepackage{xcolor}

\usepackage{subcaption}
\usepackage[margin=1in]{geometry}
\usepackage{graphicx}
\usepackage{comment}
\usepackage{enumitem}
\setlist[itemize]{leftmargin=*}
\usepackage{float}

\theoremstyle{plain}

\newtheorem{theorem}{Theorem}[section]
\newtheorem{lemma}[theorem]{Lemma}

\newtheorem{proposition}[theorem]{Proposition}
\newtheorem{assumption}[theorem]{Assumption}

\theoremstyle{remark}

\numberwithin{equation}{section}

\def\spacingset#1{\renewcommand{\baselinestretch}%
{#1}\small\normalsize} \spacingset{1}
\spacingset{1.45}

\allowdisplaybreaks

\title[Learning Preferences from the Past]{A Statistical Framework \\ for Learning Preferences from the Past} 

\author[T. Sadhukhan]{Tamojit Sadhukhan}
\address{
    Theoretical Statistics and Mathematics Unit  \\
   Indian Statistical Institute \\
    203 B. T. Road, Kolkata 700108 \\
    West Bengal, India
}
\email{tamojit96sadhukhan@gmail.com}
\author[M. Banerjee]{Moulinath Banerjee}
\address{
    Department of Statistics  \\
    University of Michigan \\
    275, West Hall, 1085 South University, Ann Arbor \\
    MI 48109, United States }
\email{moulib@umich.edu}
\author[K. Maulik]{Krishanu Maulik}
\address{
    Theoretical Statistics and Mathematics Unit \\
    Indian Statistical Institute \\
    203 B. T. Road, Kolkata 700108 \\
    West Bengal, India 
}
\email{kmisical@gmail.com}
\author[P. Roy]{Parthanil Roy}
\address{
    Department of Mathematics \\
   Indian Institute of Technology Bombay \\
    Maharashtra 400076 \\
    India 
}
\email{parthanil.roy@gmail.com}
\thanks{The first author is partially supported by an IMU Breakout Graduate Fellowship.}

\begin{document}

\begin{abstract}
In many real-world settings such as online recommendation or consumer choice modeling, individuals make repeated choices from a fixed set of options. Accurately estimating their underlying preferences is essential for generating personalized future recommendations. Probabilistic models for understanding user choice behavior from past decisions can serve as a valuable addition to existing recommender systems and choice prediction methods. To this end, in this article, we introduce a novel statistical framework for predicting user preferences based on their past choices, under a natural monotonicity assumption: options that were chosen more frequently or more intensely in the past are more likely to be chosen again in the future.  Our approach builds on a parametric model proposed by \cite{MR3657899}, originally used to describe how ants in an ant colony select a path among many pre-existing paths. We propose a non-parametric generalization of this model, drawing inspiration from the generalized elephant random walk introduced by \cite{gerw}. We develop a method of maximum likelihood estimation of the user preference probabilities under the above-mentioned monotonicity constraint. We also derive theoretical guarantees for our estimator and demonstrate the effectiveness of our method through both simulated experiments and real-world datasets. 
\end{abstract}



\maketitle

\section{Introduction}\label{sec1}

In many everyday scenarios, individuals are faced with multiple options and make choices based on their personal preferences. 
For instance, a shopper needs to decide which product to buy from a wide catalog in a retail shop or in an e-commerce platform. A viewer must choose the next movie, show, or song to watch on a streaming service. Users confronted with online ads may click on some but ignore others, while a customer ordering food using food delivery apps must pick which restaurant or dish best matches their tastes. In these scenarios, what matters is not only which option a user selects, but also how that choice is expressed. 
The intensity of a choice often provides a stronger quantification of preference than frequency alone. For example, in e-commerce, the amount a user spends on a product can reflect the strength of their commitment; in streaming platforms, the duration of viewing captures how engaged the user is with the content; in online advertising, the time spent interacting with an advertisement goes beyond a simple click to reveal genuine interest; and in food delivery apps, the size or value of an order can signal stronger loyalty than a casual, low-cost purchase. 

Across all these diverse settings, 
user behavior exhibits a natural reinforcement pattern: options that have been chosen more frequently or more intensely in the past are expected to be more likely to be chosen again in the future. Classical examples of such reinforcement dynamics can also be found in nature, for example in the foraging behavior of ants. Ants coordinate path selection through pheromone deposition: when an ant travels along a path, it leaves pheromones that increase the attractiveness of that path for subsequent ants. 
\cite{deneubourg1990self} introduced a parametric model that studied how ants in an ant colony select between two branches of a path leading to food. In this model, each ant deposits pheromone when it traverses a branch, and the probability that subsequent ants choose a branch increases with the quantity of pheromone already present in it. This local reinforcement mechanism produces a feedback loop, eventually leading the colony to favor one branch over the other, even if the two branches are initially identical. This biological example provides a natural analogy to human decision-making: 
just as pheromones bias future ant choices, past user choices can reinforce certain options, shaping future preferences in a predictable manner.

Accurately estimating user preferences is essential for entities that offer choice options, in order to deliver personalized and effective recommendations, as it directly influences user satisfaction and engagement. One approach to such preference estimation, building on the framework of \cite{deneubourg1990self}, was proposed in \cite{MR3657899}, where statistical procedures were developed for estimating the parameters of the underlying model. Specifically, \cite{MR3657899} studied a two-colored urn process as a mathematical analogue of the branch selection problem, and developed both maximum likelihood estimation and weighted least squares estimation methods. They established consistency and asymptotic normality of these estimators, providing a statistically rigorous foundation for analyzing such reinforcement dynamics. Furthermore, they applied their methods to experimental data on ant path selection, showing how model parameters could be reliably inferred from observed choices, thereby moving beyond heuristic curve-fitting methods of \cite{deneubourg1990self} that lack statistical guarantees. This line of work illustrates how reinforcement-driven behaviors can be both modeled and estimated quantitatively, offering valuable insights for analogous problem in human choice modeling.

Inspired by the work of \cite{MR3657899}, we introduce a novel framework for estimating the probabilities of a user selecting each of the given options,  based on their past choices, under a natural monotonicity assumption: the chance of an option being chosen increases with the frequency and intensity of its past selection, capturing the intuition that repeated intense choices signal stronger preferences. 
Our framework is probabilistic, interpretable, and grounded in reinforcement dynamics, offering both statistical validity and intuitive insights into user behavior. Moreover, while \cite{MR3657899} deals with a parametric setting, our approach is non-parametric, providing greater flexibility to capture diverse user behaviors without imposing strong functional form assumptions. This non-parametric approach is motivated by the generalized elephant random walk models of \cite{gerw}, with which it shares a natural connection that we elaborate on in Section~\ref{sec2}. Importantly, our model also incorporates the intensity of choices, acknowledging that not all user actions carry the same weight. Furthermore, we develop estimation procedures for this model and demonstrate their effectiveness through both theoretical analysis and empirical evaluation. 

Beyond its standalone interest, our framework can also serve as a useful 
component within modern recommendation systems. Recommendation systems have become the standard approach for predicting user preferences. Collaborative filtering 
such as neighborhood-based approach, matrix factorization, content-based filtering (for an overview of these methods, see, for example, \cite{Aggarwal2016} and the references therein) and more recently, deep learning methods (\cite{chen2022sequential}, \cite{covington2016deep}, \cite{zhang2019deep}) have achieved considerable success in practice. Probabilistic reinforcement-based models such as ours can potentially provide additional interpretable preference structures or informative priors for such systems. In particular, the ability to model calibrated choice probabilities may help enrich recommendation pipelines in settings where understanding user decision behavior is important. 

The remainder of the paper is organized as follows. In Section~\ref{sec2}, we introduce our model and describe the associated estimation methods, along with its connection to the generalized elephant random walks. Section~\ref{sec3} presents the main theoretical results and discusses their implications. Section~\ref{sec4} and Section~\ref{sec5} provide empirical evidence through simulation studies and real data applications, respectively. Proofs of the theoretical results are in Section~\ref{sec6}.

\section{Estimation of User Preference Probabilities \\ under Monotonicity Constraint}\label{sec2}

\subsection{Setup and The Model}\label{sec2.1} We observe the choice dynamics of a population of $N$ users over a fixed observation window. Each user makes a sequence of selections from a common set of $m = 2$ options, say $c_1$ and $c_0$, and these selections are recorded at discrete time points. For $j = 1, \ldots, N$, 
the ${j}$-th user makes ${T_j+1}$ choices during the observation window, where $T_j$ are independent and identically distributed, with common distribution $T$ being positive integer valued.  For each user $j = 1, \ldots, N$, we observe $\left(U^{(j)}_{t+1}, W^{(j)}_{t+1}\right)$ at time $t+1$, $t = 0, \ldots, T_j$, where 
\[
   U^{(j)}_{t+1} = \begin{cases} 1,  \quad j^{th} \text{ user chooses } c_1 \text{ at time } t+1, \\ 0, \quad j^{th} \text{ user chooses } c_0 \text{ at time } t+1,\end{cases}
\]
and $W^{(j)}_{t+1}$ is the intensity of the choice made by the $j$-th user at time $t+1$. We assume that the users are identical and independent, i.e., there exists random variables $U_{1}, U_2, \ldots $ and $W_{1}, W_2, \ldots $ such that for $j = 1, \ldots, N$, we have
\[\left(T_{j}, U^{(j)}_{1}, W^{(j)}_{1}, \ldots, U^{(j)}_{T_j+1}, W^{(j)}_{T_j+1}\right) \stackrel{iid}{\sim} \left(T, U_{1}, W_1, \ldots, U_{T+1}, W_{T+1}\right).\] For $j = 1, \ldots, N$, the relative intensity with which $c_1$ is chosen by the $j$-th user till time $t$ is denoted by $R^{(j)}_{t}$, given by,
\[
     R^{(j)}_{t} := \frac{\sum_{i=1}^{t} U^{(j)}_iW^{(j)}_i}{\sum_{i=1}^{t} W^{(j)}_i}, \quad 1 \leq t \leq T_j.
\]
Thus, for $j = 1, \ldots, N$, $\left(T_{j}, R^{(j)}_{1}, \ldots, R^{(j)}_{T_j}\right) \stackrel{iid}{\sim} \left(T, R_{1}, \ldots, R_{T}\right)$, where
\[
   R_{t} := \frac{\sum_{i=1}^{t} U_iW_i}{\sum_{i=1}^{t} W_i}, \quad 1 \leq t \leq T. 
\]
For $j = 1, \ldots, N$, our model characterizes the conditional probability of choosing $c_1$  by the $j$-th user at time $t + 1$, conditioned on $R^{(j)}_{t}$, in the following way: 
\[{\mathbb{P}\left(U^{(j)}_{t+1} = 1 \hspace{.1cm}\big|\hspace{.1cm} R^{(j)}_{t}\right)} = h(R^{(j)}_t), \quad t \geq 1, \quad \mathbb{P}\left(U^{(j)}_{1} = 1 \right) = q,\]
where $h:[0,1] \to [0,1]$ is a non-decreasing function, which we call the preference function, and $q \in (0,1)$. Our interest lies in learning the preference function $h$.

Our modeling approach is inspired in part by the (one-dimensional) generalized elephant random walk, introduced in \cite{gerw}. Elephant random walk is a one-dimensional discrete-time random walk which moves along the integer line $\mathbb{Z}$,
one step ($\pm1$) at a time. It is different than the simple symmetric random walk in the sense that it uses the memory
of the past steps to determine its future steps. In particular, the probability of taking a particular type of 
step ($+1$ or $-1$) at a time-point is a linear function of the proportion of steps of that type till that time point. The generalized elephant random walk replaces this linear structure of the classical elephant random walk with general functional forms (see, for example, Section 2 of \cite{gerw}). Analogously, our framework employs a flexible non-parametric structure to model reinforcement in user choices. When the choice intensities are equal, the sequence of choices in our model has a natural one-to-one correspondence with the steps of the generalized elephant random walk, thereby establishing a direct connection with the walk.
In particular, if, in our model, we take the preference function $h$ to be of the form \begin{align}\label{eq:gerwf}h(x) = pf(x) + (1-p)(1-f(x))\end{align} for some non-decreasing $f : [0,1] \to [0,1]$ and $p \in (1/2,1)$, and take $W_t \equiv 1$ for all $t \geq 1$, then we have the following one-to-one correspondence: \[
\left\{(2U_{t}-1) , t \geq 1\right\} \stackrel{d}{=} \left(X_t, t \geq 1\right),\] where $X_t$ is the $t$-th step of 
the generalized elephant random walk with parameters $f$ and $p$ (see, for example, Section 2.1 of \cite{gerw} for more details).

\subsection{Estimation Methods}

Our goal is to learn the function \(h\) under the natural constraint that \(h\) is non-decreasing. In particular, we seek both a {point estimate} of \(h(r_0)\) for any \(r_0\in(0,1)\) and a {confidence set} for \(h(r_0)\), which provides a 
quantification of estimation uncertainty. Confidence sets offer a principled way to assess the reliability of the point estimates of the preference probabilities, making it useful in applications where conservative decisions are desirable.

The problem of estimating monotone functions has a long history in statistics, beginning with the interval censoring or {current status model} (\cite{MR73895}), where the nonparametric maximum likelihood estimator of the underlying monotone distribution function (of, for example, the (random) time to infection) was introduced. 
Subsequent work by \cite{groeneboom1992information} established the nonstandard asymptotic theory of such estimators, highlighting $n^{1/3}$ convergence rates and limit distributions linked to the distribution of the location of the minimum of two-sided Brownian motion with quadratic drift. 
Parallel developments in {monotone density estimation}, where the maximum likelihood estimate of the unknown non-increasing density is the Grenander estimator (\cite{grenander1956theory}, \cite{GroeneboomP}), revealed similar nonstandard asymptotics. 
These ideas were later extended to monotone regression problems (\cite{brunk1969estimation}), panel count data (\cite{wellner2000two}), and mixed case interval censoring models (\cite{schick2000consistency}, \cite{song2004estimation}). A significant advance was made by \cite{banerjee2001likelihood}
who introduced the likelihood ratio test and the corresponding confidence intervals as a general tool for inference under monotonicity constraints and showed that the likelihood ratio statistic admits a pivotal asymptotic distribution, thereby enabling confidence interval construction through the standard
inversion technique. The computations of both the non-parametric maximum likelihood estimates and the likelihood ratio statistics in these contexts are generally based on the theory of generalized
isotonic regression and can be efficiently carried out using suitably adapted versions of the Pool Adjacent Violators Algorithm (\cite{robertson1988order}).
Collectively, this body of work established both the computational tractability of estimators of monotone functions and their distinctive asymptotic behavior, laying the groundwork for more recent developments in likelihood-based inference under monotonicity constraints.

We employ the nonparametric maximum likelihood estimation method based on the observations $\left(U^{(j)}_{t+1}, W^{(j)}_{t+1}\right)$, $t=0, \ldots, {T_j}$, $j = 1, \ldots, N$. 
The likelihood 
$L_N(h)$ is given by
\begin{align*}
    L_N(h) &= 
    q^N\prod_{j=1}^N \prod_{t_j = 1}^{T_{j}}\left\{h\left(R^{(j)}_{t_j}\right)\right\}^{U^{(j)}_{t_j+1}}\left\{1-h\left(R^{(j)}_{t_j}\right)\right\}^{1-U^{(j)}_{t_j+1}},
\end{align*}
and the log-likelihood $l_N(h) = \log L_N(h)$, up to a constant independent of $h$, is given by
\[
   l_N(h) = \sum_{j=1}^N\sum_{t=1}^{T_j}\left[U^{(j)}_{t+1}\psi\left(R^{(j)}_{t}\right) - \log\left(1+e^{\psi\left(R^{(j)}_{t}\right)}\right)\right], \quad
\text{where} \quad \psi(x) = \log\left(\frac{h(x)}{1-h(x)}\right)
.\]
Let $R_{(1)} < \ldots < R_{(S)}$ be the ordered enumeration of the distinct elements of $R^{(j)}_{t}$, $t=1, \ldots, {T_j}$, $j = 1, \ldots, N$. For $1 \leq s \leq S$, define
\[
   {F}_s = \sum_{j=1}^N\sum_{t=1}^{T_j}\mathbf{1}\left\{R^{(j)}_{t} = R_{(s)}\right\}, 
   \quad \overline{U}_s = \frac{1}{{F}_s}\sum_{j=1}^N\sum_{t=1}^{T_j}U^{(j)}_{t+1}\mathbf{1}\left\{R^{(j)}_{t} = R_{(s)}\right\}.
\]
The log-likelihood can alternatively be written as
\begin{align}\label{eq:ll}
  l_N(h) = \sum_{s=1}^S \left[\psi\left(R_{(s)}\right){F}_s\overline{U}_s - {F}_s\log\left(1+e^{\psi\left(R_{(s)}\right)}\right)\right].
\end{align}
Our nonparametric maximum likelihood estimator $\hat{h}_N$ of $h$ is defined as the (unique) non-decreasing right-continuous step-function, with possible jumps only at the $R_{(s)}$'s, such that the
the log-likelihood in \eqref{eq:ll} is maximized. Since only the values \(h\left(R_{(s)}\right), 1 \leq s\leq S\) are identifiable, the particular choice of $\hat{h}_N$, among equivalent representations, is arbitrary and does not affect the asymptotic results. 
By  
the Kuhn-Tucker Theorem, 
\(\hat{h}_N\left(R_{(1)}\right) \leq \ldots \leq \hat{h}_N\left(R_{(s)}\right)\) can be expressed as the solution of the following convex optimization problem:  
\begin{align}\label{eq:opt}
\min_{0 \leq h_1 \leq \ldots \leq h_S \leq 1}\sum_{s=1}^S {F}_s(\overline{U}_s - h_s)^2.
\end{align}
To introduce the solution of \eqref{eq:opt}, we need the following well-known notation. For a sequence of points \(\{(r_0,u_0) = (0,0),(r_1,u_1),\ldots,(r_l, u_l)\}\) with \(r_0 < r_1 < \cdots < r_l\), define the left-continuous piecewise-constant function \(g(r)\) such that \(g(r_i) = u_i\), $i = 0, \ldots, l$, with \(g(r)\) being constant on each \((r_{i-1},r_i)\) and denote by $\operatorname{slogcm}\{r_i,u_i\}_{i=0}^l$ 
the vector of slopes (left-derivatives) of the greatest convex minorant of \(g(r)\), 
evaluated at the points \((r_1,\ldots,r_l)\). Then it is known that (see, for example, \cite{robertson1988order}), the solution of \eqref{eq:opt} is given by
\begin{align}\label{eq:hath}
    \left\{\hat{h}_N\left(R_{(s)}\right)\right\}_{s = 1}^{S} = \operatorname{slogcm}\left\{\sum_{a=1}^s{F}_a, \sum_{a=1}^s{F}_a\overline{U}_a\right\}_{s = 0}^{S}.
\end{align}
Here we are using the convention that summation over an empty set is $0$. Thus the point estimate $\hat{h}_N(r_0)$ of $h(r_0)$, for any $r_0 \in (0,1)$, is given by the following: get $s_0$ with ${R_{(s_0)} \leq r_0 < R_{(s_0+1)}}$ and define $\hat{h}_N\left(r_0\right) := \hat{h}_N\left(R_{(s_0)}\right)$.

The confidence set for \(h(r_0)\) is obtained by inverting the likelihood-ratio tests for testing the null hypotheses of the form ${H_0 : h(r_0) = h_0}$, $0 \leq h_0 \leq 1$. Under ${H_0 : h(r_0) = h_0}$, 
the estimator $\hat{h}^{(0)}_N$ of $h$ is defined as the (unique) non-decreasing right-continuous step-function, with possible jumps only at the $R_{(s)}$'s and $r_0$, such that the
the log-likelihood in \eqref{eq:ll} is maximized with the additional condition that $h(r_0) = h_0$. 
As before, 
\(\hat{h}^{(0)}_N\left(R_{(1)}\right) \leq \ldots \leq \hat{h}^{(0)}_N\left(R_{(s)}\right)\) can be expressed as the solution of the following convex optimization problem: 
\begin{align}\label{eq:opt2}
\min_{0 \leq h_1 \leq \ldots \leq h_{s_0} \leq h_0 \leq h_{s_0+1} \leq \ldots \leq h_S \leq 1}\sum_{s=1}^S {F}_s(\overline{U}_s - h_s)^2,
\end{align}
where $s_0$ is such that ${R_{(s_0)} \leq r_0 < R_{(s_0+1)}}$. 
The solution of \eqref{eq:opt2} is given by
\begin{align}\label{eq:hath01}
    \left\{\hat{h}^{(0)}_N\left(R_{(s)}\right)\right\}_{s = 1}^{s_0} = h_0 \wedge \operatorname{slogcm}\left\{\sum_{a=1}^s{F}_a, \sum_{a=1}^s{F}_a\overline{U}_a\right\}_{s = 0}^{s_0},
\end{align}
while
\begin{align}\label{eq:hath02}
    \left\{\hat{h}^{(0)}_N\left(R_{(s)}\right)\right\}_{s = s_0+1}^{S} = h_0 \vee \operatorname{slogcm}\left\{\sum_{a=s_0+1}^s{F}_a, \sum_{a=s_0+1}^s{F}_a\overline{U}_a\right\}_{s = s_0}^{S}.
\end{align}
Here the minimum and maximum are taken component-wise. The (twice) likelihood ratio statistic for testing ${H_0 : h(r_0) = h_0}$ is then given by
\begin{align*}
    2\log \lambda_N(h_0) = 2\left(l_N\left(\hat{h}_N\right) - l_N\left(\hat{h}^{(0)}_N\right)\right).
\end{align*}
The reason behind not using the asymptotic distribution of $\hat{h}_N(r_0)$ to obtain a Wald-type confidence set for $h(r_0)$ and the usefulness of the likelihood ratio statistic based confidence set is explained in Section~\ref{sec3}.

\section{Theoretical Results}\label{sec3}
Before describing the theoretical results for our estimation methods, we first state some regularity assumptions that are standard in the monotone function estimation literature. 
\begin{assumption}\label{a1}
   The positive integer valued random variable $T$ has bounded support. In other words, there exists a $T_0 \in \mathbb{N}$ such that $1 \leq T \leq T_0$ almost surely.
\end{assumption}
\begin{assumption}\label{a12}
   The process $(U_t, W_t)_{t\geq 1}$ is independent of $T$.
\end{assumption}

Our main results characterize the local behavior of the estimators at a fixed point $r_0 \in (0,1)$. The following assumptions ensure regularity of the model in a neighborhood $\mathcal{N}(r_0)$ of $r_0$.

\begin{assumption}\label{a2}
The function $\mathcal{F}(r) = \mathbb{E}\left( \sum_{t=1}^T \mathbf{1}\left\{R_{t} \leq r\right\} \right)$ is strictly increasing in $\mathcal{N}(r_0)$.
\end{assumption}
\begin{assumption}\label{a3}
    The two dimensional marginal distribution functions of the process $(R_t)_{t\geq1}$ \[(r_1,r_2) \mapsto \mathbb{P}\left(R_t \leq r_1, R_s \leq r_2\right), \quad (r_1,r_2) \in [0,1]^2, \quad s,t \geq 1,\]  is twice continuously differentiable in $\mathcal{N}(r_0) \times \mathcal{N}(r_0)$. In particular, the distribution function of $R_t$ is twice continuously differentiable in $\mathcal{N}(r_0)$ with the derivative $f_t$ satisfying $f_t(r_0) \neq 0$. 
\end{assumption}
From Assumptions~\ref{a1} - \ref{a3}, it follows that $\mathcal{F}'(r_0)$ exists and is non-zero.
\begin{assumption}\label{a4}
    The function $h$ is continuously differentiable in $\mathcal{N}(r_0)$ with $h'(r_0) \neq 0$.
\end{assumption}
Our first result is about the asymptotic distribution of our estimator at $r_0$.
\begin{theorem}\label{thm1}
Under Assumptions~\ref{a1} - \ref{a4}, we have
    \[N^{1 / 3}\left(\hat{h}_{N}\left(r_{0}\right)-h\left(r_{0}\right)\right) \stackrel{d}{\rightarrow} \left(\frac{4h\left(r_{0}\right)\left(1-h\left(r_{0}\right)\right) h^{\prime}\left(r_{0}\right)}{\mathcal{F}^{\prime}\left(r_{0}\right)}\right)^{1 / 3} \underset{x \in \mathbb{R}}{\arg \max }\left(\mathbb{B}(x)-x^{2}\right),\]
    where $\mathbb{B}$ is a two-sided Brownian motion satisfying $\mathbb{B}(0) = 0$.
\end{theorem}
The distribution of the random variable \(\mathbb{Z} := {\arg \max }_{x \in \mathbb{R}}\left(\mathbb{B}(x)-x^{2}\right)\) is the Chernoff's distribution (see, for example, Chapter 3 of \cite{MR3445293}, for an introduction to Chernoff's distribution). 

From Theorem~\ref{thm1}, we can construct an Wald-type approximate \((1-\alpha)\)-level confidence interval for 
\(h(r_{0})\) as
\[
\left( \hat{h}_{N}(r_{0}) - N^{-1/3} C_{h}q_{\mathbb{Z},1-\alpha/2},\;
       \hat{h}_{N}(r_{0}) + N^{-1/3} C_{h}q_{\mathbb{Z},1-\alpha/2} \right),
\]
where \(q_{\mathbb{Z},1-\alpha/2}\) denotes the \((1-\alpha/2)\)-quantile of \(\mathbb{Z}\) 
and
\[
C_{h} = 
\left( \frac{4\hat{h}_{N}(r_{0})(1-\hat{h}_{N}(r_{0}))\hat{h}'(r_{0})}{\widehat{\mathcal{F}}'(t_{0})} \right)^{1/3},
\]
with \(\hat{h}'\) and \(\widehat{\mathcal{F}}'\) denoting estimators of \(h'\) and \(\mathcal{F}'\), respectively. 
In practice, estimating \(\mathcal{F}'\) requires first estimating the probability mass function of \(T\) 
and then that of the finite dimensional distributions of the process \((R_1, R_2, \ldots)\). 
Furthermore, estimating \(h'\) is also an additional challenge. 
This is a common issue in monotone function estimation (see, for example, \cite{banerjee2005confidence} 
for a discussion), which can be circumvented by employing the likelihood ratio statistic based confidence set.

In order to derive the asymptotic distribution of the likelihood ratio statistic, we first need to obtain the asymptotic behavior of the processes $\Omega_N$ and $\Omega^{(0)}_N$, defined by, for fix $0 < r_0 < 1$, 
\begin{align}\label{eq:omega}
    \Omega_N(x) &= N^{1/3}\left(\hat{h}_N\left(r_0 + N^{-1/3}x\right) - h(r_0)\right), \quad x \in \mathbb{R}, \nonumber \\ 
    \Omega^{(0)}_N(x) &= N^{1/3}\left(\hat{h}^{(0)}_N\left(r_0 + N^{-1/3}x\right) - h(r_0)\right), \quad x \in \mathbb{R}.
\end{align}
To characterize the asymptotic distributions of $\Omega_N$ and $\Omega^{(0)}_N$, we introduce some stochastic processes (and their associated functionals).
For $\mathbb{B}$ a two-sided standard Brownian motion satisfying $\mathbb{B}(0) = 0$ and fixed constants $a,b > 0$, define the process $\widetilde{\mathbb{B}}_{a, b}$ as 
\begin{align}\label{eq:bm}\widetilde{\mathbb{B}}_{a,b}(x) = a \mathbb{B}(x) + b x^{2}, \quad x \in \mathbb{R}.\end{align}
Denote by \(\mathbb{G}_{a,b}\) the greatest convex minorant of the process 
\(\widetilde{\mathbb{B}}_{a,b}\). The right derivative of \(\mathbb{G}_{a,b}\), written as \(g_{a,b}\), is a piecewise constant, non-decreasing function, and on every compact interval 
it has only finitely many jumps. The left-sided and right-sided restrictions of these processes are defined as follows. 
Let \(\mathbb{G}^{-}_{a,b}\) (respectively, \(\mathbb{G}^{+}_{a,b}\)) denote the greatest convex minorant of \(\widetilde{\mathbb{B}}_{a,b}\) restricted to the set \(x \leq 0\) (respectively, \(x > 0\)), and let \({g}^{-}_{a,b}\) (respectively, \(g^{+}_{a,b}\) ) be its right-derivative process.  
Finally, define the combined process 
\[
g_{a,b}^{(0)}(x) = 
\begin{cases}
g^{-}_{a,b}(x) \wedge 0, & x \leq 0, \\
g^{+}_{a,b}(x) \vee 0, & x > 0.
\end{cases}
\]
The process \(g^{(0)}_{a,b}\) shares the same properties as 
\(g_{a,b}\): it is a piecewise constant, non-decreasing function with only finitely 
many jumps on any compact interval. 
on a finite interval containing zero. 
For detailed accounts of these processes, we refer to \cite{banerjee2001likelihood}. 

Define $\mathcal{L}_2$ (respectively, $\mathcal{L}_{\infty}$) to be the space of real-valued functions on $\mathbb{R}$ which are bounded on every compact set, equipped with the topology of $L_2$-convergence with respect to Lebesgue measure on compact sets (respectively, uniform convergence on compact sets). The following result establishes the asymptotic behavior of $\Omega_N$ and $\Omega^{(0)}_N$.

\begin{theorem}\label{thm2} Under Assumptions~\ref{a1} - \ref{a4}, we have
\begin{align*}
    \left(\Omega_N, \Omega^{(0)}_N\right) {\to} \left(g_{\alpha,\beta}, g^{(0)}_{\alpha,\beta}\right), \quad \text{ where } \quad 
    \alpha^2 = \frac{h\left(r_{0}\right)\left(1-h\left(r_{0}\right)\right)}{\mathcal{F}^\prime\left(r_{0}\right)}, \quad \beta = \frac{1}{2}h'(r_0),
\end{align*}
where the convergence is 
in the space $\mathcal{L}_2 \times \mathcal{L}_2$.
\end{theorem}

The asymptotic distribution of the likelihood ratio statistic is given by the following result.
\begin{theorem}\label{thm:lrt} Let Assumptions~\ref{a1} - \ref{a4} hold. Under the null hypothesis $H_0 : h(r_0) = h_0$,
\begin{align*}
2 \log \lambda_{N}(h_0)
& \stackrel{d}{\rightarrow} \mathbb{D} \;:=\; \int \Big\{ (g_{1,1}(x))^2 - (g^{0}_{1,1}(x))^2 \Big\}\, dx.
\end{align*}
\end{theorem}

Using Theorem~\ref{thm:lrt}, it is straightforward to construct likelihood ratio based 
confidence sets for \(h(r_{0})\). 
An asymptotic \((1-\alpha)\)-level confidence set for \(h(r_{0})\) is given by
\[
\big\{h_0 : 2 \log \lambda_{n}(h_0) \leq q_{\mathbb{D},1-\alpha} \big\},
\]
where \(q_{\mathbb{D},1-\alpha}\) denotes the \((1-\alpha)\)-quantile of \(\mathbb{D}\)
(quantiles of $\mathbb{D}$ are well-known and can be found, for example, in \cite{banerjee2001likelihood}). 
The construction of confidence sets in this framework reduces to computing the likelihood 
ratio statistic across a family of null hypotheses. This computation is straightforward and 
can be carried out using appropriate versions of the Pool Adjacent Violators Algorithm. 
The estimation of nuisance parameters is the central difficulty in constructing the Wald-type 
confidence intervals. 
By contrast, the likelihood ratio based method avoids nuisance parameter estimation entirely and 
thus provides a remarkably clean  
procedure for constructing confidence intervals 
for \(h(r_{0})\), making it the more appealing choice.  

\section{Simulation}\label{sec4}

\subsection{Simulation Setup}\label{sec:ss}

To evaluate the performance of our estimator, we conduct a series of simulation studies. The data-generating mechanism is based on the strictly increasing preference function
\[
h(r) \;=\; \tfrac{2}{5} r^{2} + \tfrac{3}{10}, \quad r \in [0,1].
\]
This choice of preference function is motivated by the fact that it corresponds to the one-dimensional generalized elephant random walk with $p = 7/10$ and $f(x) = x^2$ (see \eqref{eq:gerwf} and the discussion at the end of Section~\ref{sec2.1}).
Each simulated dataset consists of  
$N$ users. 
We vary the number of 
users across the settings
$N \in \{300, 600, 900, 1200, 1500\}$,
to examine finite-sample behavior as the sample size increases. We consider two settings for the number of choices ($T+1$) made by each user: 

\begin{enumerate}[nosep]
    \item {Constant:} \(T \sim \delta_{20}\) i.e., $T$ is degenerate at $20$. 
    \item {Random (Truncated Poisson):} 
    For $X \sim \operatorname{Poisson} (20)$, $T$ is given by
    \[
       T = \begin{cases}4, &X \leq 4, \\[-.4cm] X, &4 < X < 20,  \\[-.4cm] 20, &20 \leq X.\end{cases}
    \]
\end{enumerate}
The intensity of a user’s choice at each setup is modeled under three regimes:
\begin{enumerate}[nosep]
    \item {Equal intensity:} 
    \(W_{t} \overset{iid}{\sim} \delta_1\), i.e., all choices have equal weights degenerate at $1$.
    \item {Uniform intensity:} 
    \(W_{t} \overset{iid}{\sim} U(0,1)\), i.e., the strength of choices vary 
    uniformly over $(0,1)$. 
    \item Persistent intensity:
For \(U_{t} \overset{iid}{\sim} U(0,1)\), $(W_t)_{t \geq 1}$ is given by 
    \[
        W_1 = U_1, \quad W_t = \begin{cases} W_{t-1}, &\text{ with probability } 0.2,\\[-.4cm] U_t, &\text{ with probability } 0.8,\end{cases} \quad t \geq 2.
    \]
\end{enumerate}
We use \(\lfloor T/2 \rfloor N\) observations, corresponding to the second through $(\lfloor T/2 \rfloor + 1)$-th choices for each user, to train the estimator. The remaining \(\lfloor T/2 \rfloor\) 
choices for each user play the role of test data to assess the estimator's performance. Each experiment was independently replicated {100 times} to assess variability and 
ensure reliable comparisons across settings.

\subsection{Evaluation Metrics}

The last \(\lfloor T/2 \rfloor\)  choices and the corresponding predictor covariates 
of each user are reserved for testing. The resulting test dataset is
\[
\big\{ \big(U^{(j)}_{t+1}, R^{(j)}_{t}\big) : t = \lfloor T/2 \rfloor + 1, \ldots, T,\; j=1,\ldots,N \big\},
\]
consisting of \(\lfloor T/2 \rfloor N\) pairs of observations. To assess the performance of the proposed estimator, we consider two
evaluation metrics: the {test error}, which measures predictive accuracy, 
and the {expected calibration error (ECE)}, which evaluates the calibration 
of the estimated probabilities of the user choices.

\subsubsection{Test Error}

To compute the test  error, we compare $\hat{h}_N$ with an \emph{oracle estimate} ${h}_{oracle}$ of the true preference 
function $h$. The oracle estimate is a Nadaraya–Watson type estimate of $h$, obtained using the \emph{unobserved} test data in the following way:
\[
{h}_{oracle}(r)
:= \frac{ \sum_{j=1}^{N}\sum_{t=\lfloor T/2 \rfloor + 1}^{T}
   U^{(j)}_{t+1}\,\phi\left(\frac{r - R^{(j)}_{t}}{\zeta} \right)}{\sum_{j=1}^{N}\sum_{t=\lfloor T/2 \rfloor + 1}^{T}
          \phi\left(\frac{r - R^{(j)}_{t}}{\zeta}\right)}, \quad r \in (0,1),
\]
where, $\phi$ is the standard normal density and $\zeta$ is a bandwidth parameter. We set $\zeta = 0.01$ throughout our experiments. The test error is defined as the average absolute deviation between the estimated 
preference function \(\hat{h}_N\) and the oracle estimate $h_{oracle}$, evaluated over a discrete grid $S_R := \left\{\frac{x}{99} : x  = 0, \ldots, 99\right\}$ of 100 equidistant values of $r$ in $[0,1]$: 
\[
\text{Test Error} \;=\; \frac{1}{|S_R|} \sum_{r \in S_R} 
\big| {h}_{oracle}(r) - \hat{h}_N(r) \big|.
\]
Intuitively, this metric compares the predicted probabilities of the user choices
with the empirical (weighted) frequencies 
observed in the test data. 
Thus the oracle estimator serves as a benchmark constructed using the unobserved test data, against which we compare the predictions of the monotone estimator $\hat{h}_N$ obtained from the training data. One may also consider a monotone version of the oracle estimator; however, in our experiments we use the unconstrained oracle estimate in order to benchmark against a less restricted estimate obtained from the test data.

\subsubsection{Expected Calibration Error (ECE)}

To quantify calibration, we compute the 
{expected calibration error (ECE)} following standard reliability 
analysis methods used in probabilistic prediction (see, for example, \cite{PakdamanNaeini_Cooper_Hauskrecht_2015} and the references therein). The predicted probabilities \(\hat{h}_N\left(R^{(j)}_{t}\right)\), 
for \(t = \lfloor T/2 \rfloor + 1, \ldots, T,\; j=1,\ldots,N\),  are grouped into \(B = 50\) equal-width bins on \([0,1]\). For each bin \(b \subset [0,1]\), let 
\[
   I_b := \big\{ t, j : t \in \{\lfloor T/2 \rfloor + 1, \ldots, T\},\; j \in \{1,\ldots,N\},\; R^{(j)}_{t} \in b\big\},
\]
and $n_b := |I_b|$. Further let
\[
c_b := \frac{1}{n_b} \sum_{(t,j) \in I_b } \hat{h}_N\left(R^{(j)}_{t}\right)
\quad \text{and} \quad 
a_b := \frac{1}{n_b} \sum_{(t,j) \in I_b } U^{(j)}_{t+1}
\]
denote the average predicted probability ({confidence}) and the average empirical 
frequency ({accuracy}), respectively. The expected calibration error (ECE) is then given by
\[
\mathrm{ECE} = \sum_{b=1}^{B}  \frac{n_b}{\lfloor T/2 \rfloor N}
\, \big| c_b - a_b \big|.
\]
Smaller ECE-values indicate better-calibrated probability 
estimates.

\subsection{Simulation Results}

We evaluate the performance of our estimator under all six configurations arising from the two settings for the number of choices and the three choice-intensity regimes described in Section~\ref{sec:ss}. For each configuration, we compute the mean test error and the mean expected calibration error (ECE), averaged over 100 replicated  experiments and across the considered sample sizes~\(N\). These results, along with the corresponding standard deviations, are summarized in Table~\ref{tab:sim1}. 
We also report the average (across 100 replicated experiments) length and coverage probability (empirical proportion of replications in which the confidence set contains the true value) of the $95\%$ confidence set for \(h(1/3)\) in Table~\ref{tab:sim2}.

In Table \ref{tab:sim1}, the mean test error, as well as the corresponding standard deviation, decreases steadily with increasing sample size $N$ in every configurations. At the same time, the mean expected calibration error (ECE) and the associated standard deviation also declines consistently, indicating that the estimated preference probabilities become increasingly well-calibrated. Notably, the mean expected calibration errors (ECE) are already small even at moderate sample sizes. Table \ref{tab:sim2} shows that the obtained confidence intervals exhibit desirable behavior. The average length of the confidence sets decreases monotonically with $N$, indicating improved precision. Importantly, this gain in precision does not come at the expense of validity. The coverage probabilities remain close to the nominal $0.95$ across all sample sizes and configurations, with no systematic evidence of under-coverage.

Taken together, these findings demonstrate that the proposed estimator provides accurate and well-calibrated estimates of preference probabilities with valid confidence intervals. Furthermore, larger sample sizes improve estimation precision and yield tighter confidence sets while maintaining appropriate coverage.

\begin{table}[htbp]
\centering
\renewcommand{\arraystretch}{0.75}

\begin{subtable}{0.48\textwidth}
\setlength{\tabcolsep}{10pt}
\begin{tabular}{ccc}
\toprule
$N$ & Test Error & ECE \\
\midrule
300  & 0.062 (0.012) & 0.024 (0.008) \\
600  & 0.043 (0.007) & 0.019 (0.006) \\
900  & 0.035 (0.005) & 0.016 (0.005) \\
1200 & 0.031 (0.005) & 0.014 (0.004) \\
1500 & 0.028 (0.004) & 0.014 (0.004) \\
\bottomrule
\end{tabular}
\centering
\caption{\centering Configuration 1: equal intensity \\ and constant number of choices}
\end{subtable}
\hfill
\begin{subtable}{0.48\textwidth}
\setlength{\tabcolsep}{10pt}
\begin{tabular}{ccc}
\toprule
$N$ & Test Error & ECE \\
\midrule
300  & 0.061 (0.01) & 0.027 (0.008) \\
600  & 0.042 (0.006) & 0.02 (0.005) \\
900  & 0.035 (0.006) & 0.017 (0.005) \\
1200 & 0.031 (0.005) & 0.015 (0.004) \\
1500 & 0.029 (0.004) & 0.015 (0.004) \\
\bottomrule
\end{tabular}
\centering
\caption{\centering Configuration 2: equal intensity \\ and random number of choices}
\end{subtable}

\vspace{0.4cm}

\begin{subtable}{0.48\textwidth}
\setlength{\tabcolsep}{10pt}
\begin{tabular}{ccc}
\toprule
$N$ & Test Error & ECE \\
\midrule
300  & 0.051 (0.007) & 0.026 (0.007) \\
600  & 0.037 (0.005) & 0.021 (0.005) \\
900  & 0.03 (0.004) & 0.018 (0.004) \\
1200 & 0.027 (0.003) & 0.016 (0.004) \\
1500 & 0.024 (0.003) & 0.014 (0.004) \\
\bottomrule
\end{tabular}
\centering
\caption{\centering Configuration 3: uniform intensity \\ and constant number of choices}
\end{subtable}
\hfill
\begin{subtable}{0.48\textwidth}
\setlength{\tabcolsep}{10pt}
\begin{tabular}{ccc}
\toprule
$N$ & Test Error & ECE \\
\midrule
300  & 0.052 (0.007) & 0.029 (0.008) \\
600  & 0.038 (0.005) & 0.022 (0.006) \\
900  & 0.031 (0.004) & 0.019 (0.005) \\
1200 & 0.027 (0.004) & 0.016 (0.004) \\
1500 & 0.025 (0.003) & 0.015 (0.004) \\
\bottomrule
\end{tabular}
\centering
\caption{\centering Configuration 4: uniform intensity \\ and random number of choices}
\end{subtable}

\vspace{0.4cm}

\begin{subtable}{0.48\textwidth}
\setlength{\tabcolsep}{10pt}
\begin{tabular}{ccc}
\toprule
$N$ & Test Error & ECE \\
\midrule
300  & 0.053 (0.008) & 0.026 (0.009) \\
600  & 0.037 (0.005) & 0.02 (0.005) \\
900  & 0.031 (0.005) & 0.017 (0.005) \\
1200 & 0.027 (0.004) & 0.015 (0.004) \\
1500 & 0.025 (0.004) & 0.014 (0.004) \\
\bottomrule
\end{tabular}
\centering
\caption{\centering Configuration 5: persistent intensity \\ and constant number of choices}
\end{subtable}
\hfill
\begin{subtable}{0.48\textwidth}
\setlength{\tabcolsep}{10pt}
\begin{tabular}{ccc}
\toprule
$N$ & Test Error & ECE \\
\midrule
300  & 0.053 (0.007) & 0.028 (0.009) \\
600  & 0.038 (0.005) & 0.022 (0.006) \\
900  & 0.031 (0.004) & 0.018 (0.005) \\
1200 & 0.028 (0.003) & 0.016 (0.004) \\
1500 & 0.025 (0.003) & 0.015 (0.004) \\
\bottomrule
\end{tabular}
\centering
\caption{\centering Configuration 6: persistent intensity \\ and random number of choices}
\end{subtable}
\centering
\caption{\centering Simulation results under six configurations: mean test error and mean expected calibration error (ECE) (standard deviation within parentheses) across different sample sizes.}
\label{tab:sim1}
\end{table}

\begin{table}[htbp]
\centering
\renewcommand{\arraystretch}{0.75}
\begin{subtable}{0.48\textwidth}
\setlength{\tabcolsep}{10pt}
\begin{tabular}{ccc}
\toprule
$N$ & Length & Coverage\\
\midrule
300  & 0.087 & 0.94\\
600  & 0.069 & 0.93 \\
900  & 0.062 & 0.96 \\
1200 & 0.056 & 0.99 \\
1500 & 0.053 & 0.94 \\
\bottomrule
\end{tabular}
\centering
\caption{\centering Configuration 1: equal intensity \\ and constant number of choices}
\end{subtable}
\hfill
\begin{subtable}{0.48\textwidth}
\setlength{\tabcolsep}{10pt}
\begin{tabular}{ccc}
\toprule
$N$ & Length & Coverage\\
\midrule
300  & 0.094 & 0.96\\
600  & 0.078 & 0.97\\
900  & 0.068 & 0.94\\
1200 & 0.062 & 0.94\\
1500 & 0.059 & 0.95\\
\bottomrule
\end{tabular}
\centering
\caption{\centering Configuration 2: equal intensity \\ and random number of choices}
\end{subtable}

\vspace{0.4cm}

\begin{subtable}{0.48\textwidth}
\setlength{\tabcolsep}{10pt}
\begin{tabular}{ccc}
\toprule
$N$ & Length & Coverage\\
\midrule
300  & 0.084 & 0.98\\
600  & 0.065 & 0.98\\
900  & 0.058 & 0.94\\
1200 & 0.052 & 0.95\\
1500 & 0.049 & 0.95\\
\bottomrule
\end{tabular}
\centering
\caption{\centering Configuration 3: uniform intensity \\ and constant number of choices}
\end{subtable}
\hfill
\begin{subtable}{0.48\textwidth}
\setlength{\tabcolsep}{10pt}
\begin{tabular}{ccc}
\toprule
$N$ & Length & Coverage\\
\midrule
300  & 0.085 & 0.93\\
600  & 0.068 & 0.98\\
900  & 0.059 & 0.95\\
1200 & 0.052 & 0.92\\
1500 & 0.049 & 0.98\\
\bottomrule
\end{tabular}
\centering
\caption{\centering Configuration 4: uniform intensity \\ and random number of choices}
\end{subtable}

\vspace{0.4cm}

\begin{subtable}{0.48\textwidth}
\setlength{\tabcolsep}{10pt}
\begin{tabular}{ccc}
\toprule
$N$ & Length & Coverage \\
\midrule
300  & 0.081 & 0.95\\
600  & 0.062 & 0.97\\
900  & 0.056 & 0.95\\
1200 & 0.051 & 0.97\\
1500 & 0.047 & 0.95\\
\bottomrule
\end{tabular}
\centering
\caption{\centering Configuration 5: persistent intensity \\ and constant number of choices}
\end{subtable}
\hfill
\begin{subtable}{0.48\textwidth}
\setlength{\tabcolsep}{10pt}
\begin{tabular}{ccc}
\toprule
$N$ & Length & Coverage \\
\midrule
300  & 0.083 & 0.95\\
600  & 0.065 & 0.95\\
900  & 0.056 & 0.91\\
1200 & 0.052 & 0.96\\
1500 & 0.048 & 0.96\\
\bottomrule
\end{tabular}
\centering
\caption{\centering Configuration 6: persistent intensity \\ and random number of choices}
\end{subtable}
\centering
\caption{\centering Simulation results under six configurations: $95\%$ confidence set length and coverage probability for $h(1/3)$ across different sample sizes.}
\label{tab:sim2}

\end{table}

\section{Real Data Analysis}\label{sec5}

One potential application of our framework arises in streaming platforms, where reliably estimating user preference probabilities may help identify genres or types of content that align strongly with individual viewing behavior. Motivated by this application, we apply the proposed method to the MovieLens (20M) dataset (\cite{10.1145/2827872}), a widely used movie recommendation benchmark that records user ratings for movies across different genres. 
Generated on October 17, 2016, it contains 20000263 ratings, created by 138493 users between January 09, 1995 and March 31, 2015, across 27278 movies. All selected users had rated at least 20 movies. 
In our framework, movie genres play the role of options. 

In each specification of the model, we consider a pair of genres (or, a pair of groups of multiple genres) and restrict attention to movies belonging to these two groups, as well as to users who have viewed (and rated) at least 20 movies in total within the selected pair. Each such movie viewed (and rated) by a user is interpreted as a choice between the two groups under consideration. Movies that are classified as belonging to both groups are excluded from the analysis, and if a user rated the same movie more than once, only the earliest rating is retained. These restrictions ensure that each observation corresponds to an unambiguous binary choice and that there is sufficient within-user variation. The pairs of groups (of genres) considered and the corresponding number of \emph{valid} users in each model specification are reported in Table~\ref{tab:models}. 

\begin{table}[h]
\centering
\renewcommand{\arraystretch}{0.75}
\begin{tabular}{|c|c|c|c|}
\hline
Model & {Option $\boldsymbol{c_1}$} & {Option $\boldsymbol{c_2}$} & {Number of Valid Users} \\
\hline
(a) & Action  & Romance & 89752\\
(b) & Comedy  & Romance & 76094\\
(c) & Action, Adventure, Thriller & Romance, Drama & 112038 \\
(d) & Action, Adventure, Thriller & Comedy & 112340 \\
\hline
\end{tabular}
\caption{Model specifications.}
\label{tab:models}
\end{table}

We consider two different choice intensity specification, allowing us to assess whether modeling richer intensity information leads to improved predictive performance, 
\begin{enumerate}[nosep]
\item \textit{Constant intensity:} all choices receive equal weight.
\item \textit{Rating-based intensity:} The intensity of each choice is given by the rating assigned by the user to the corresponding movie.
\end{enumerate}

For evaluation, the dataset constructed for each specification is split into training and testing subsets. Specifically, for $N$, the number of users in a given model specification, the model is trained on $5N$ observations corresponding to the 11th through 15th choices for each user and tested on the 
$5N$ observations corresponding to the 16th through 20th choices for each user. The first 10 choices are excluded from the analysis to reduce noise associated with initial, potentially unstable user behavior. The same evaluation metrics as in the simulation studies, namely the test error and the empirical classification error (ECE) are employed. 

Table~\ref{tab:real} summarizes the results across four model specifications under two intensity choices. These results suggest that incorporating rating-based intensities does not lead to substantial improvements over the constant intensity specification in this dataset. Both approaches perform similarly in terms of predictive accuracy and calibration, indicating that the simpler constant-intensity model may already capture the essential structure of the data. We also provide a visual comparison of these findings: Figure~\ref{fig:error} plots the estimated preference function $\hat{h}_N$ against the oracle estimate $h_{oracle}$, while Figure~\ref{fig:ece} presents reliability plots (predicted probabilities versus observed empirical frequencies) with the diagonal line as a reference.

To illustrate the use of the estimates, we present in Table~\ref{tab:real2} the estimated preference function for Model (b) with constant choice intensity at five representative values of $r$, along with their corresponding confidence sets. In practice, how these estimates are utilized depends on the streaming service provider’s tolerance for risk. For example, a conservative provider may prioritize recommending movies of the comedy genre (option $\boldsymbol{c_1}$ in Model (b)) only to users whose past relative intensity of choosing this genre is $r = 82/99$, as such users exhibit a high estimated preference along with a narrow confidence interval. In contrast, users with past relative intensity $r = 49/99$ have a moderately high estimated preference for the comedy genre but a wider confidence interval, indicating greater uncertainty in the estimate. A more exploratory provider may still recommend movies of this genre to such users. 

\begin{table}[t]
\centering
\renewcommand{\arraystretch}{0.75}

\begin{subtable}{0.48\textwidth}
\setlength{\tabcolsep}{10pt}
\begin{tabular}{ccc}
\toprule
Choice Intensity & Test Error & ECE \\
\midrule
Constant  & 0.024 & 0.019\\
Rating-based  & 0.029 & 0.028\\
\bottomrule
\end{tabular}
\centering
\caption*{\centering Model (a)}
\end{subtable}
\hfill
\begin{subtable}{0.48\textwidth}
\setlength{\tabcolsep}{10pt}
\begin{tabular}{ccc}
\toprule
Choice Intensity & Test Error & ECE \\
\midrule
Constant  & 0.019 & 0.008\\
Rating-based  & 0.021 & 0.008\\
\bottomrule
\end{tabular}
\centering
\caption*{\centering Model (b)}
\end{subtable}

\vspace{0.4cm}

\begin{subtable}{0.48\textwidth}
\setlength{\tabcolsep}{10pt}
\begin{tabular}{ccc}
\toprule
Choice Intensity & Test Error & ECE \\
\midrule
Constant  & 0.026 & 0.014 \\
Rating-based  & 0.026 & 0.015 \\
\bottomrule
\end{tabular}
\centering
\caption*{\centering Model (c)}
\end{subtable}
\hfill
\begin{subtable}{0.48\textwidth}
\setlength{\tabcolsep}{10pt}
\begin{tabular}{ccc}
\toprule
Choice Intensity & Test Error & ECE \\
\midrule
Constant  & 0.031 & 0.019 \\
Rating-based  & 0.033 & 0.017 \\
\bottomrule
\end{tabular}
\centering
\caption*{\centering Model (d)}
\end{subtable}
\centering
\caption{\centering Real data analysis results: test error and expected calibration error (ECE) across different choice intensities and model specifications.}
\label{tab:real}
\end{table}

\begin{table}[t]
\centering
\renewcommand{\arraystretch}{0.75}
\setlength{\tabcolsep}{10pt}
\begin{tabular}{|c|c|c|c|}
\hline
$r$ & ${h}_{oracle}(r)$ & $\hat{h}_N(r)$ & Confidence Set\\
\hline
16/99 & 0.214 & 0.235 & (0.210, 0.301)\\
1/3 & 0.470 & 0.467 & (0.452, 0.534)\\
49/99 & 0.598 & 0.583 & (0.573, 0.611)\\
2/3 & 0.702 & 0.691 & (0.686, 0.718)\\
82/99 & 0.770 & 0.773 & (0.770, 0.778)\\
\hline
\end{tabular}
\centering
\caption{\centering Real data analysis results: Model (b) with constant intensity}
\label{tab:real2}
\end{table}

\begin{figure}[htbp]
\centering

\begin{subfigure}{0.24\textwidth}
  \includegraphics[width=\linewidth]{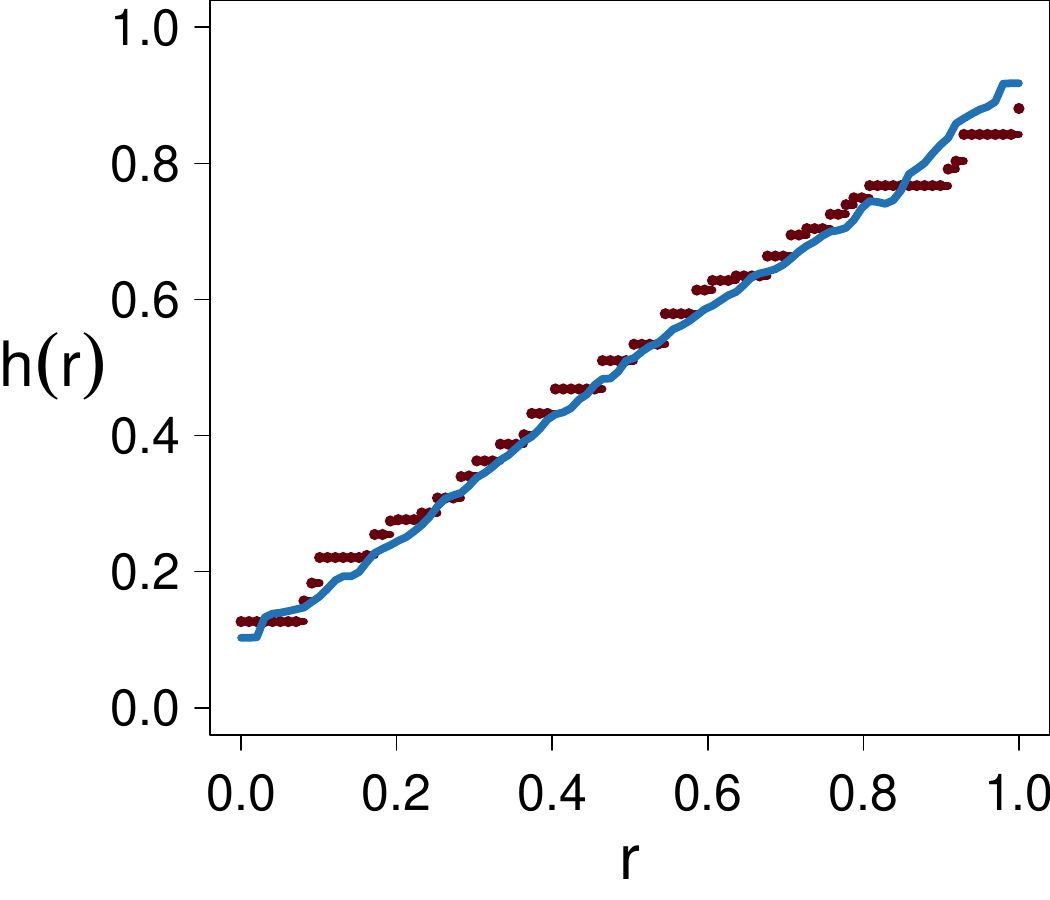}
\end{subfigure}\hfill
\begin{subfigure}{0.24\textwidth}
  \includegraphics[width=\linewidth]{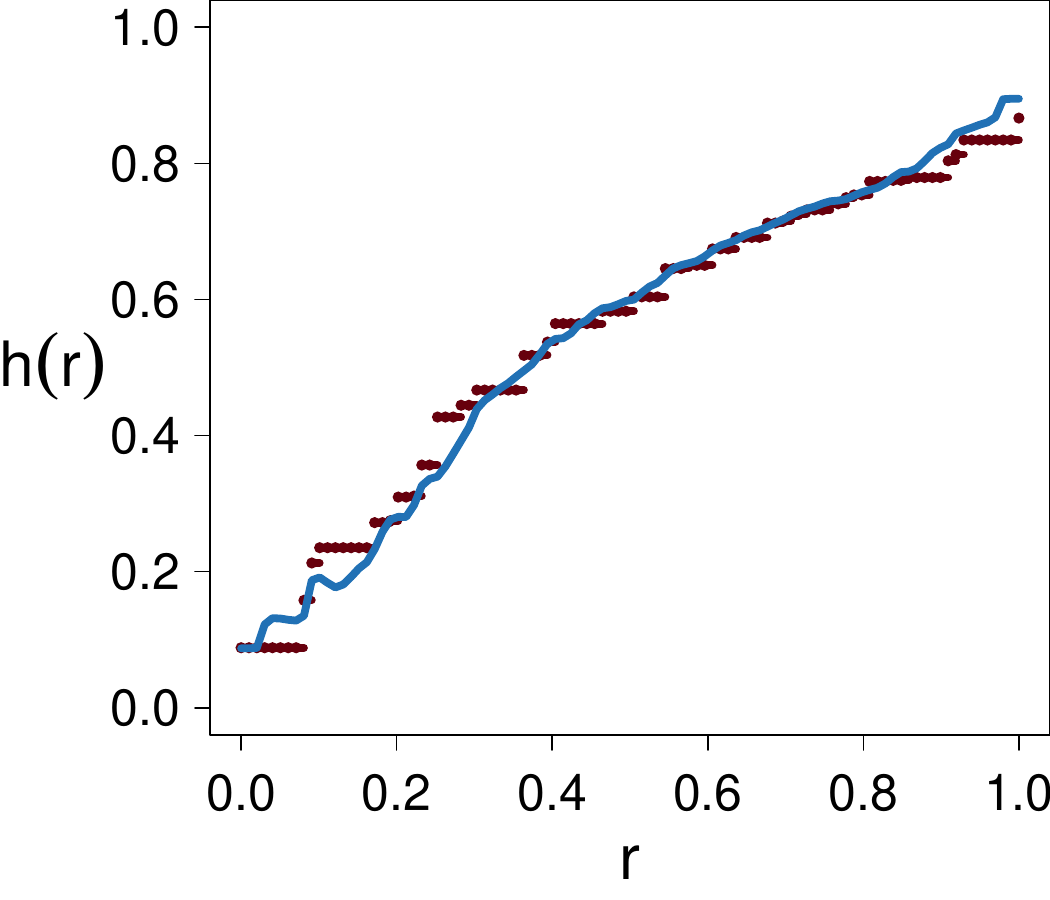}
\end{subfigure}\hfill
\begin{subfigure}{0.24\textwidth}
  \includegraphics[width=\linewidth]{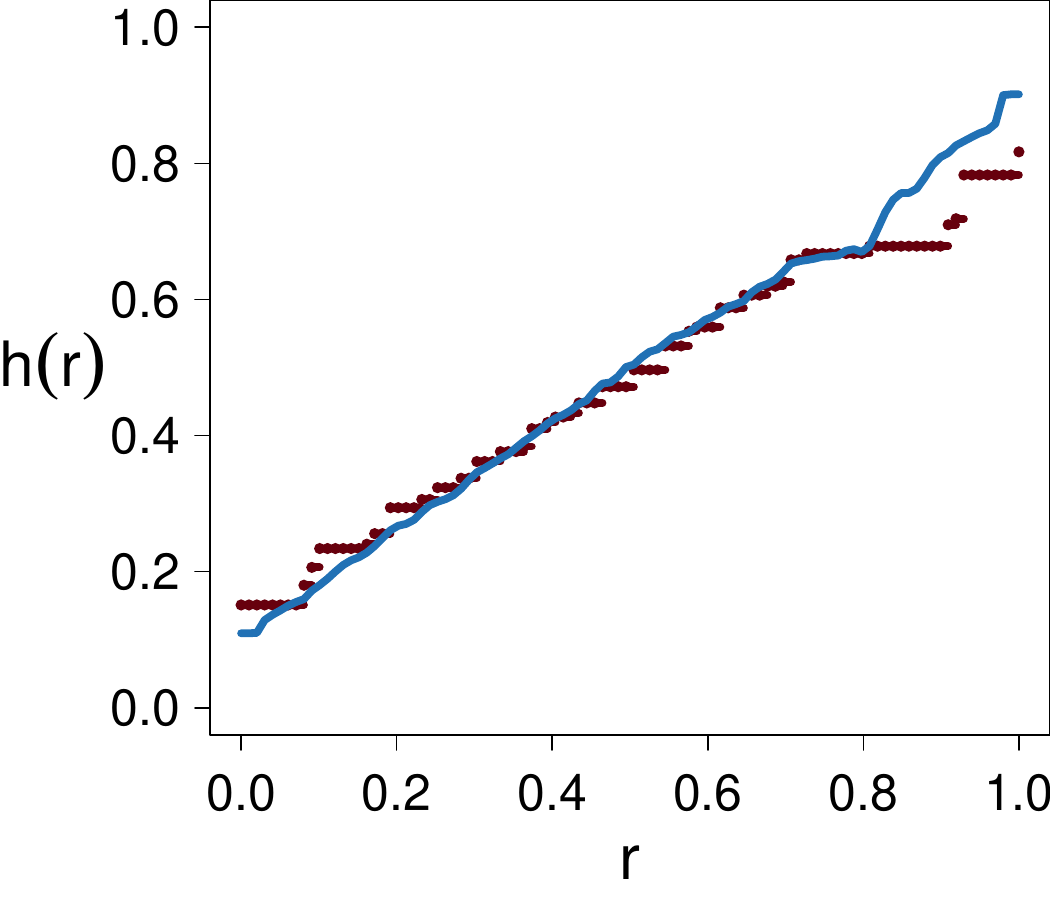}
\end{subfigure}\hfill
\begin{subfigure}{0.24\textwidth}
  \includegraphics[width=\linewidth]{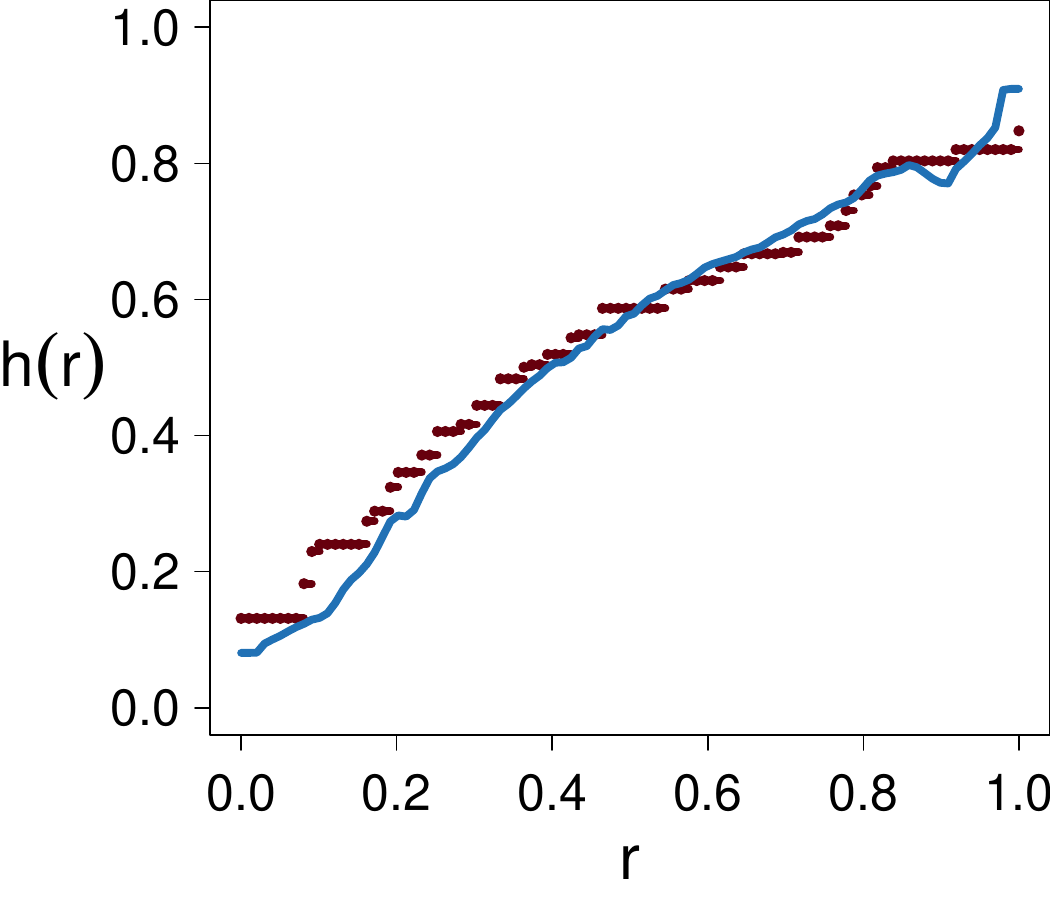}
\end{subfigure}

\vspace{0.2cm}

\begin{subfigure}{0.24\textwidth}
  \includegraphics[width=\linewidth]{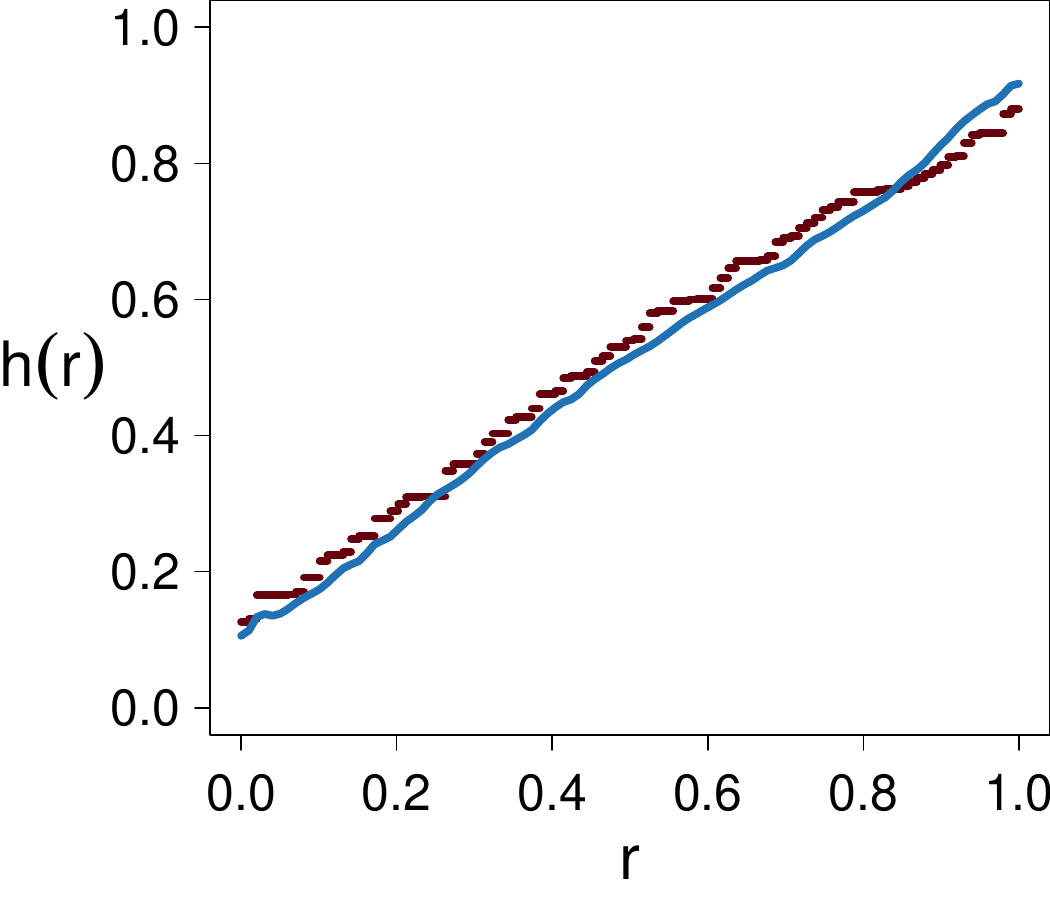}
  \caption*{\centering Model (a)}
\end{subfigure}\hfill
\begin{subfigure}{0.24\textwidth}
  \includegraphics[width=\linewidth]{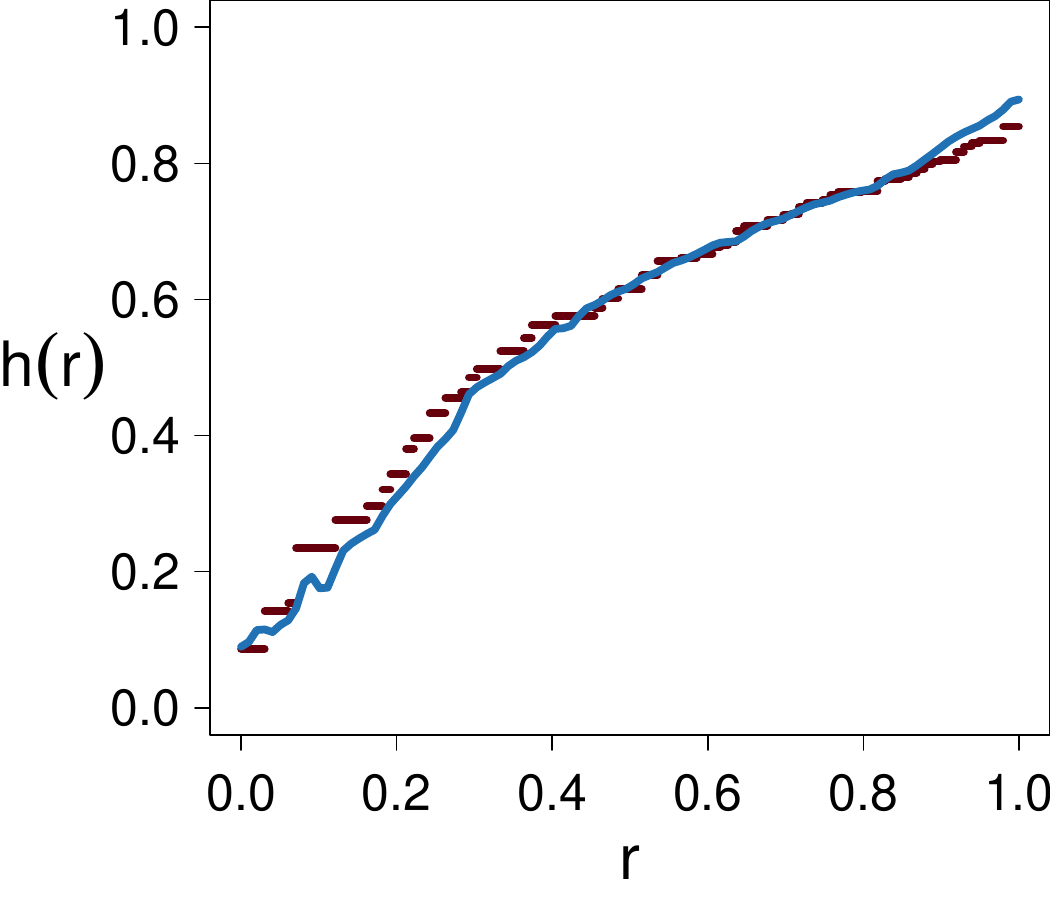}
  \caption*{\centering Model (b)}
\end{subfigure}\hfill
\begin{subfigure}{0.24\textwidth}
  \includegraphics[width=\linewidth]{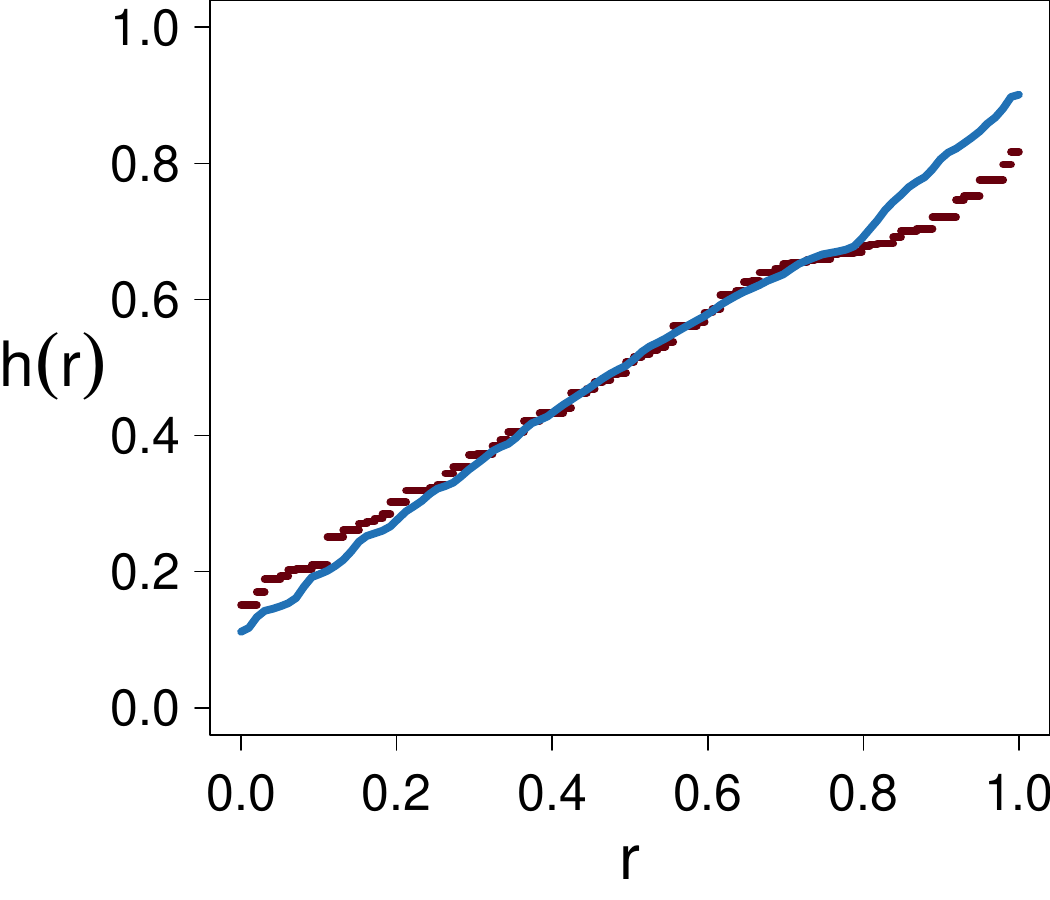}
  \caption*{\centering Model (c)}
\end{subfigure}\hfill
\begin{subfigure}{0.24\textwidth}
  \includegraphics[width=\linewidth]{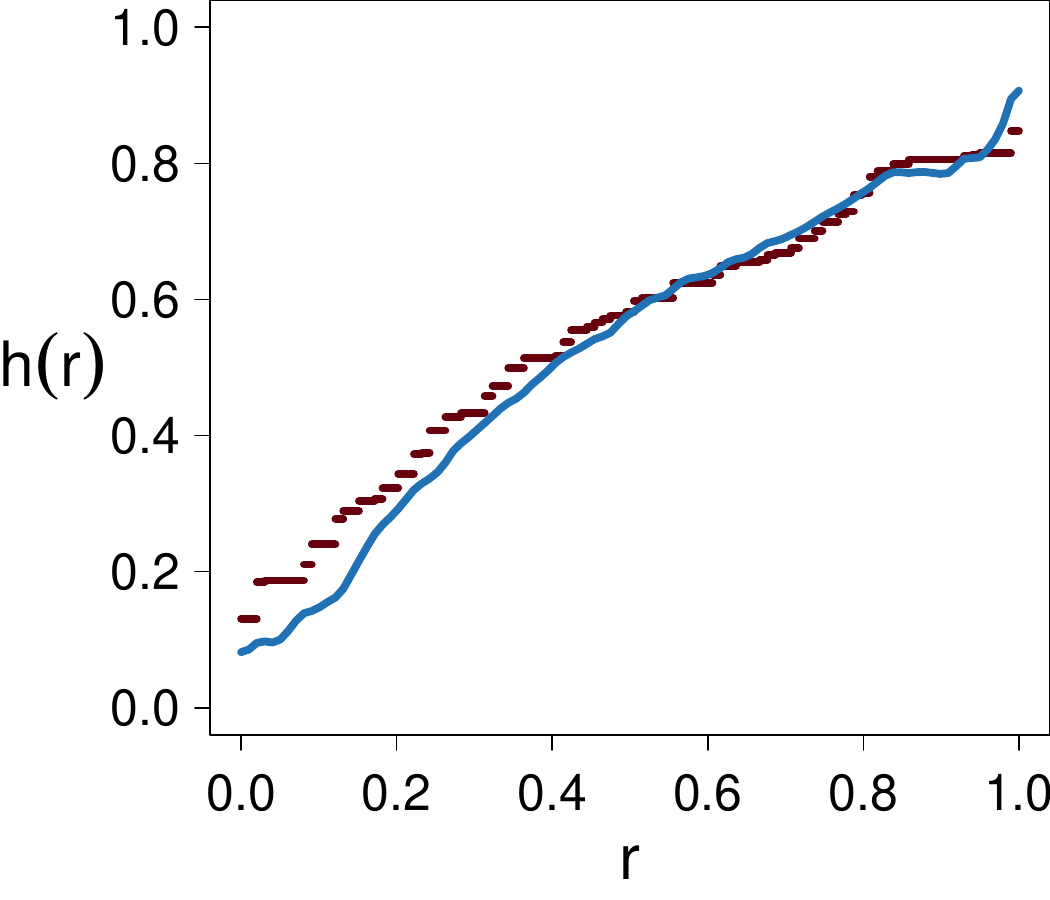}
  \caption*{\centering Model (d)}
\end{subfigure}

\vspace{0.2cm}

\begin{subfigure}{0.6\textwidth}
\centering
  \includegraphics[width=\linewidth]{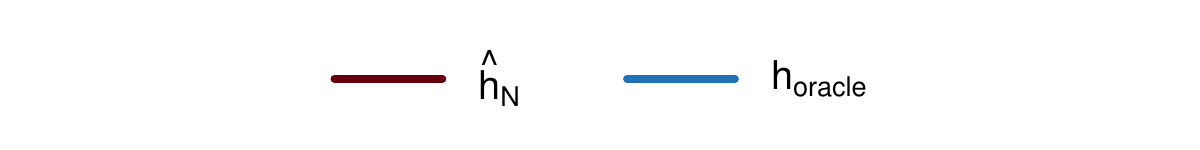}
\end{subfigure}\hfill

\caption{\centering Real data analysis results: top - constant intensity, bottom - rating-based intensity}
\label{fig:error}

\end{figure}

\begin{figure}[htbp]
\centering

\begin{subfigure}{0.24\textwidth}
  \includegraphics[width=\linewidth]{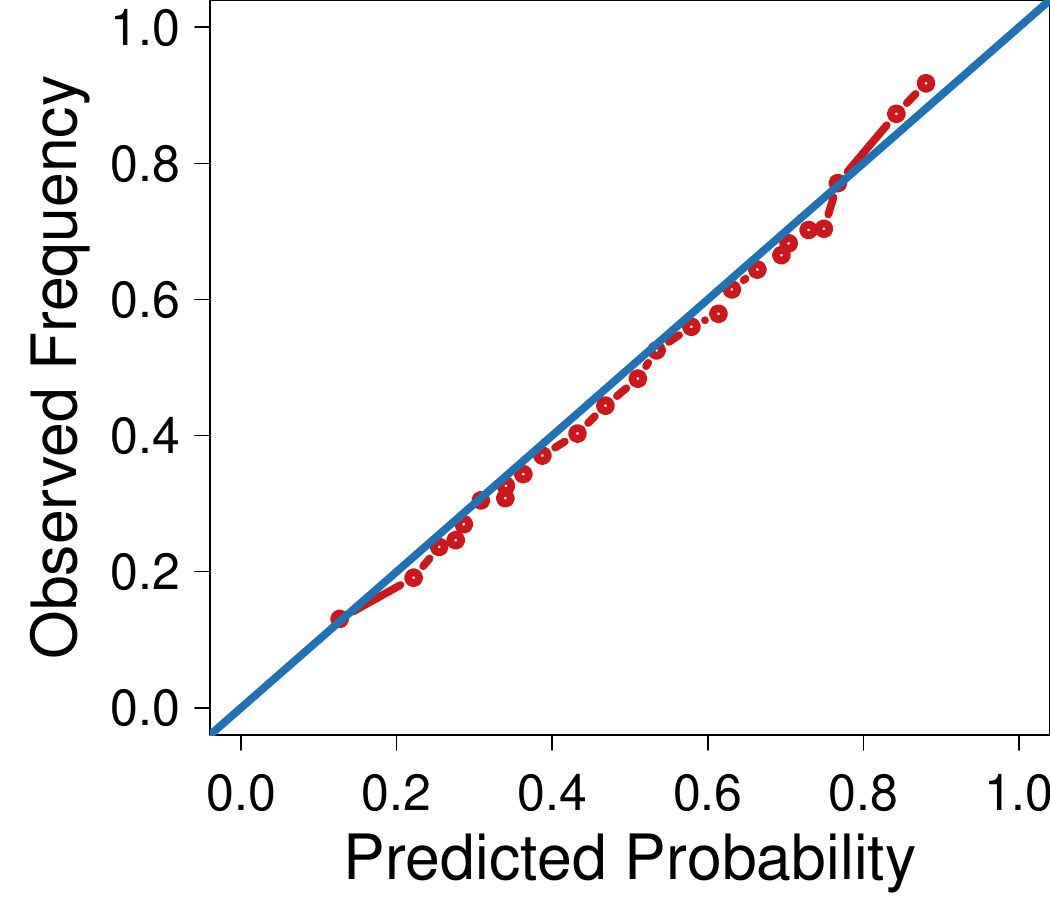}
\end{subfigure}\hfill
\begin{subfigure}{0.24\textwidth}
  \includegraphics[width=\linewidth]{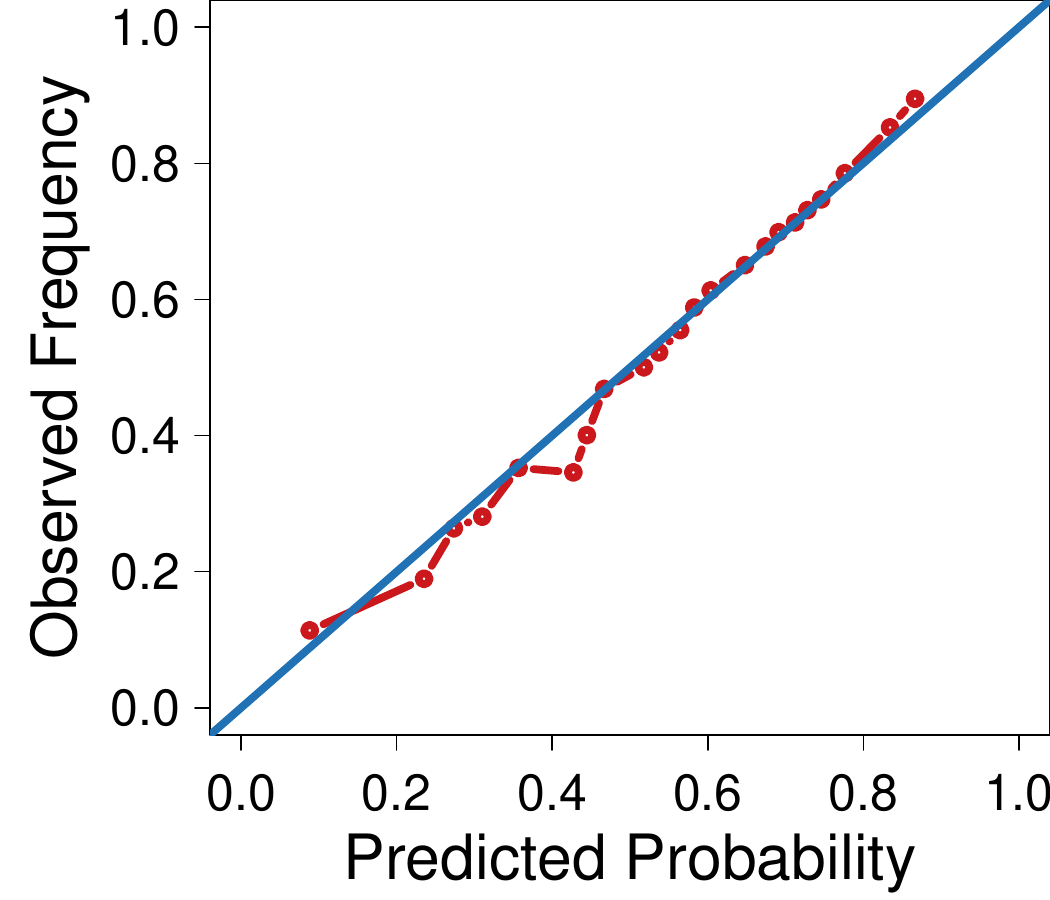}
\end{subfigure}\hfill
\begin{subfigure}{0.24\textwidth}
  \includegraphics[width=\linewidth]{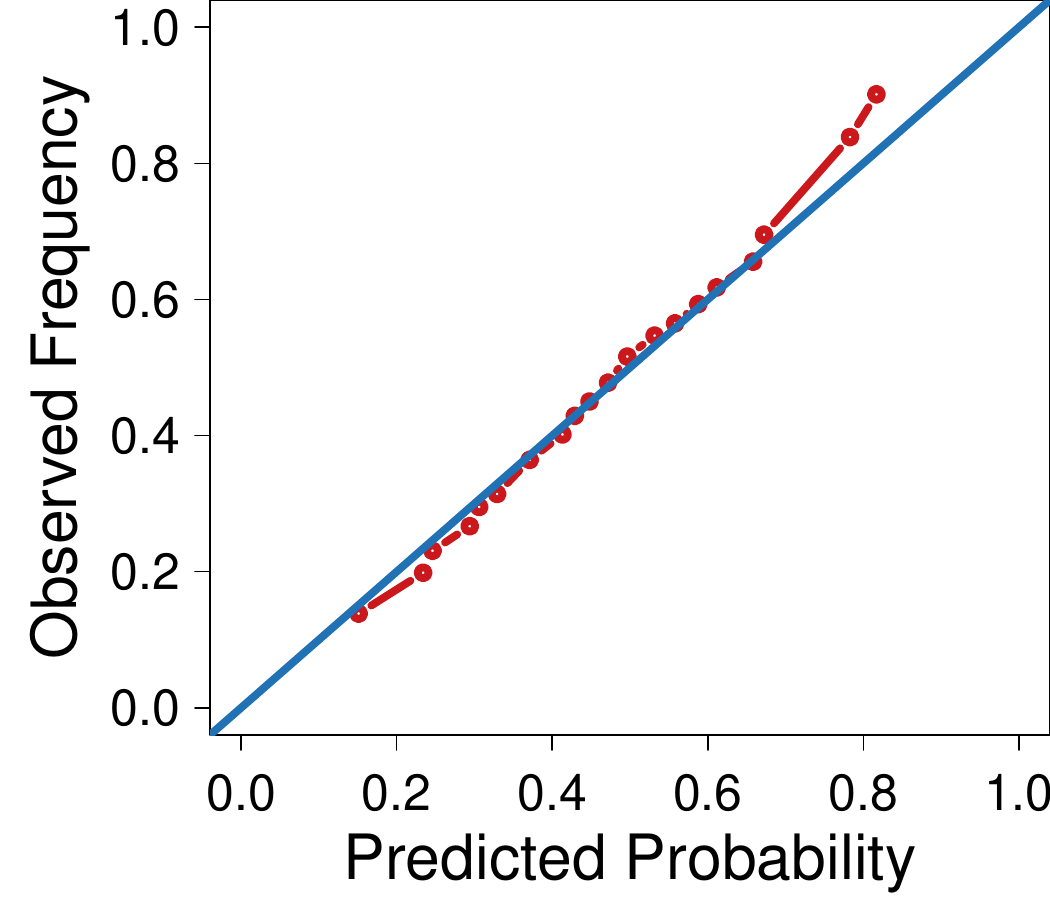}
\end{subfigure}\hfill
\begin{subfigure}{0.24\textwidth}
  \includegraphics[width=\linewidth]{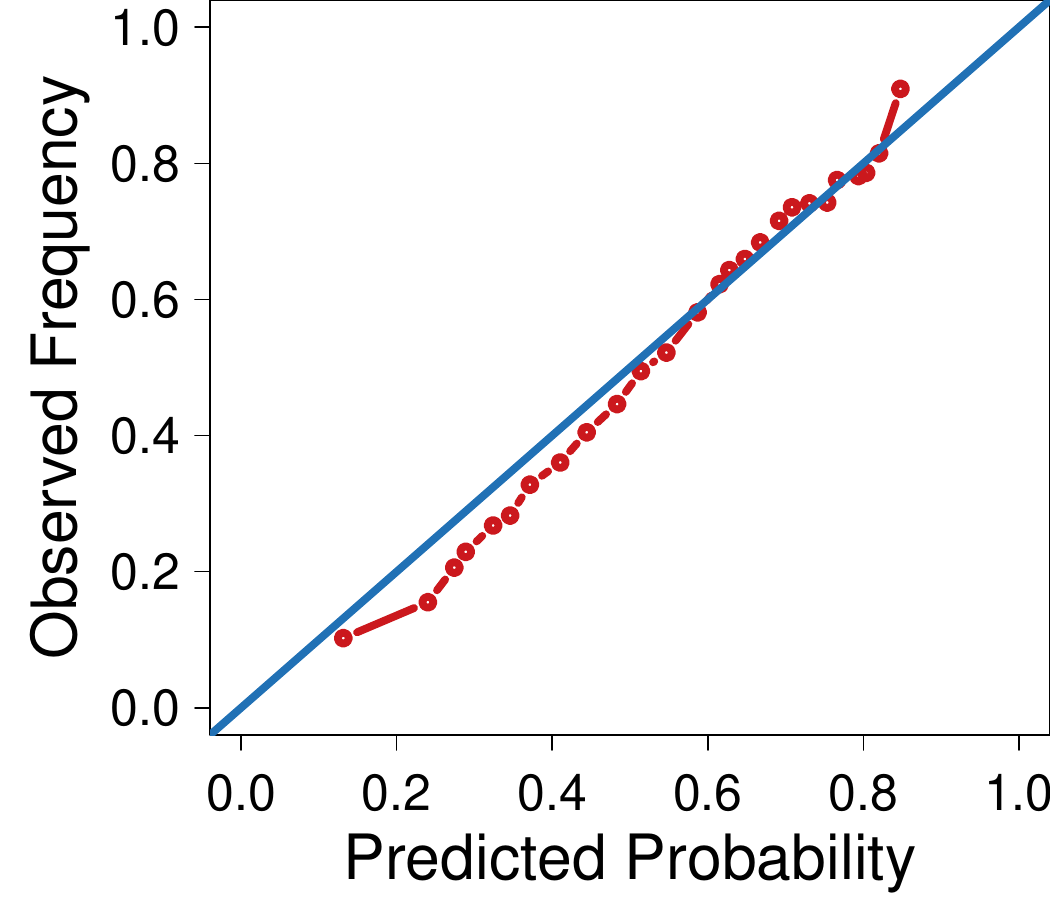}
\end{subfigure}

\vspace{0.2cm}

\begin{subfigure}{0.24\textwidth}
  \includegraphics[width=\linewidth]{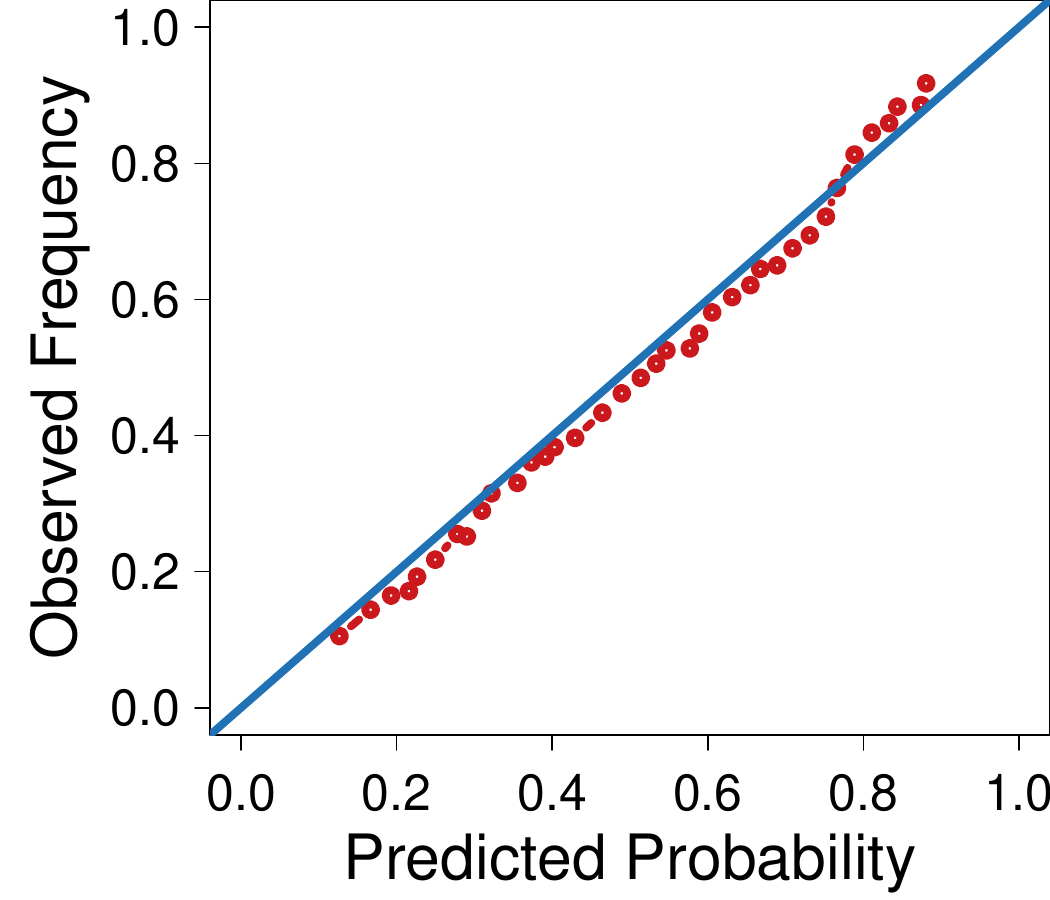}
  \caption*{\centering Model (a)}
\end{subfigure}\hfill
\begin{subfigure}{0.24\textwidth}
  \includegraphics[width=\linewidth]{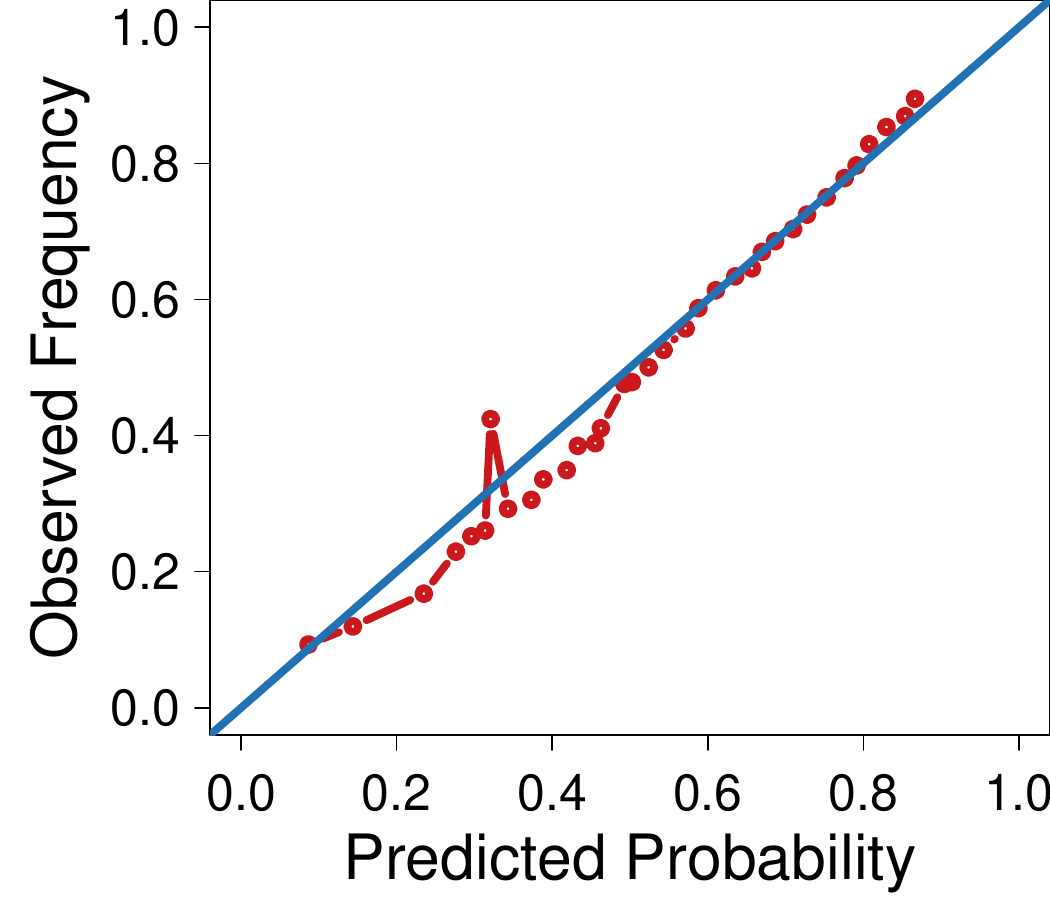}
  \caption*{\centering Model (b)}
\end{subfigure}\hfill
\begin{subfigure}{0.24\textwidth}
  \includegraphics[width=\linewidth]{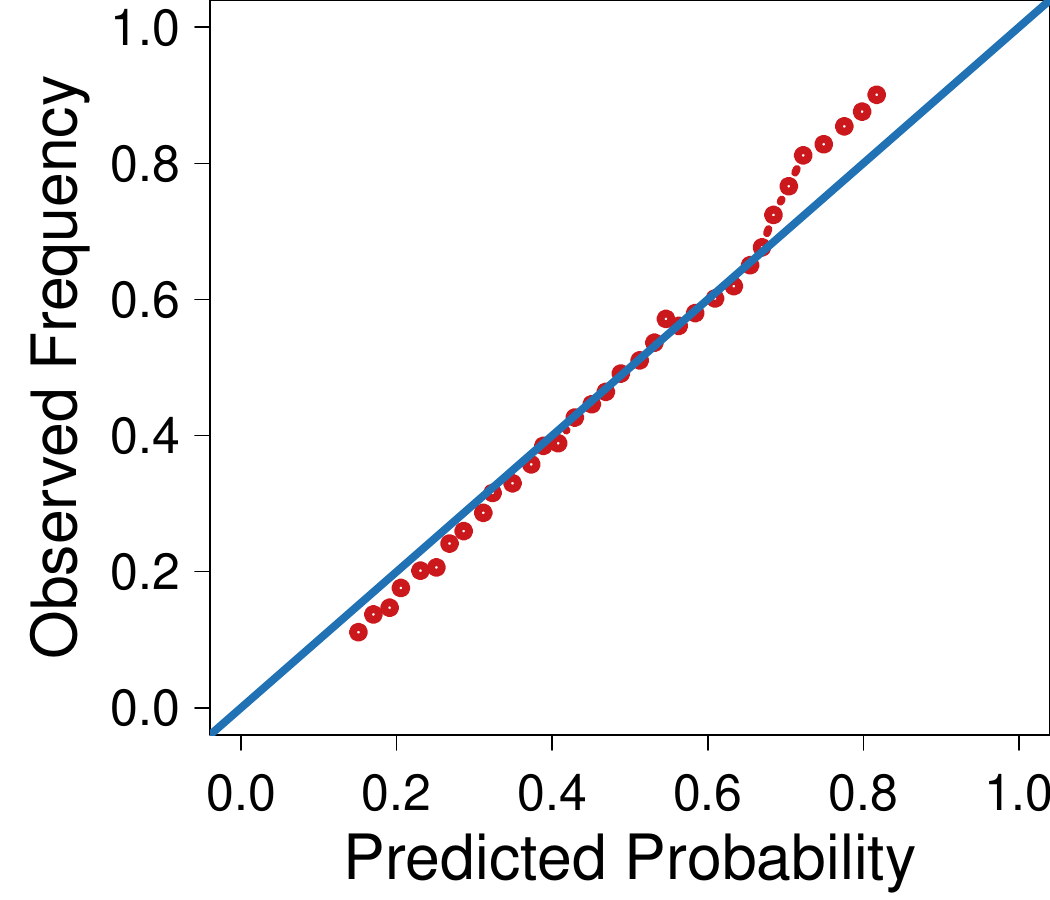}
  \caption*{\centering Model (c)}
\end{subfigure}\hfill
\begin{subfigure}{0.24\textwidth}
  \includegraphics[width=\linewidth]{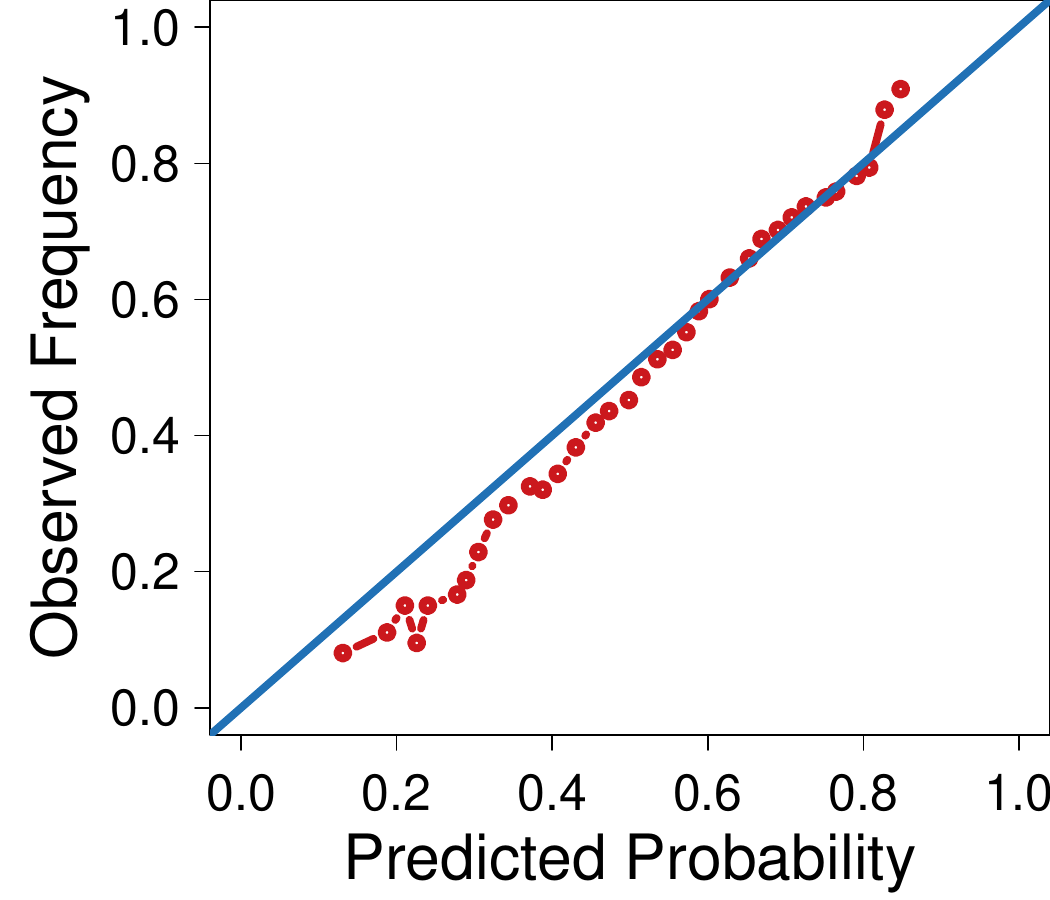}
  \caption*{\centering Model (d)}
\end{subfigure}

\vspace{0.2cm}

\begin{subfigure}{0.5\textwidth}
\centering
  \includegraphics[width=\linewidth]{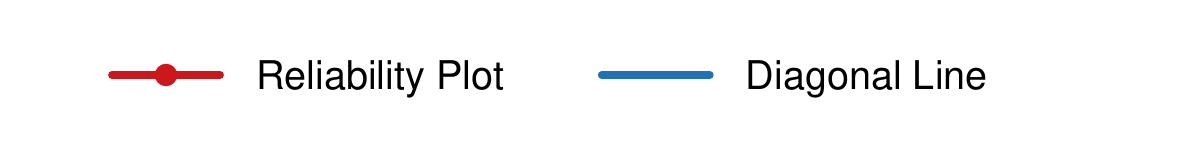}
\end{subfigure}\hfill

\caption{\centering Real data analysis results: top - constant intensity, bottom - rating-based intensity}
\label{fig:ece}

\end{figure}

\section{Conclusion}\label{sec-conc}

The proposed framework opens several directions for future research. A natural extension is to the non-stationary setting. In many real-world applications, user preferences evolve over time due to changing trends, external influences, or shifting individual interests, making it important to incorporate non-stationary dynamics into the model. Another promising direction is the inclusion of covariates, such as user-specific characteristics or contextual features associated with the available options. Incorporating such information can be expected to improve the flexibility and applicability of the framework in modeling complex decision-making behavior and will be explored in subsequent work.

\section{Proofs of the Theoretical Results}\label{sec6}
To prove our results, we first define the empirical measure $\mathbb{P}_N$, given by,
\[
    \mathbb{P}_N =  \frac{1}{N}\sum_{j=1}^N\delta_{\left\{\left(U^{(j)}_{t+1},R^{(j)}_{t}\right)_{t=1}^{T_j}\right\}},
\]
as well as the associated empirical processes
\begin{align*}
    \mathbb{F}_N(r) &= \mathbb{P}_N \sum_{t=1}^{T} \mathbf{1}\left\{R_t \leq r\right\}= \frac{1}{N}\sum_{j=1}^N\sum_{t=1}^{T_j}\mathbf{1}\left\{R^{(j)}_{t} \leq r\right\}, \\ \mathbb{G}_N(r) &= \mathbb{P}_N\sum_{t=1}^{T} U_{t+1} \mathbf{1}\left\{R_t \leq r\right\} =\frac{1}{N}\sum_{j=1}^N\sum_{t=1}^{T_j}U^{(j)}_{t+1}\mathbf{1}\left\{R^{(j)}_{t} \leq r\right\}, \quad r \in [0,1].
\end{align*}
The estimators $\hat{h}_N$, given by \eqref{eq:hath}, and $\hat{h}^{(0)}_N$, given by \eqref{eq:hath01} and \eqref{eq:hath02}, then can be written as, with the convention that $R_{(0)} = 0$,
\begin{align}\label{eq:hfg}
    \left\{\hat{h}_N\left(R_{(s)}\right)\right\}_{s = 1}^{S} &= \operatorname{slogcm}\left\{\mathbb{F}_N\left(R_{(s)}\right), \mathbb{G}_N\left(R_{(s)}\right)\right\}_{s = 0}^{S}, \nonumber\\
    \left\{\hat{h}^{(0)}_N\left(R_{(s)}\right)\right\}_{s = 1}^{s_0} &= h_0  \wedge \operatorname{slogcm}\left\{\mathbb{F}_N\left(R_{(s)}\right), \mathbb{G}_N\left(R_{(s)}\right)\right\}_{s = 0}^{s_0}, \\ \left\{\hat{h}^{(0)}_N\left(R_{(s)}\right)\right\}_{s = s_0+1}^{S} &= h_0 \vee \operatorname{slogcm}\left\{\mathbb{F}_N\left(R_{(s)}\right) - \mathbb{F}_N\left(R_{(s_0)}\right), \mathbb{G}_N\left(R_{(s)}\right) - \mathbb{G}_N\left(R_{(s_0)}\right)\right\}_{s = s_0}^{S}.\nonumber
\end{align}
Let $P$ be distribution of the random vector $\left(U_{t+1},R_{t}\right)_{t=1}^{T}$.
Note that 
\begin{align*}
\mathcal{F}(r)&=P\sum_{t=1}^{T} \mathbf{1}\left\{R_t \leq r\right\}
    = \mathbb{E}\!\left( \sum_{t=1}^{T} \mathbf{1}\{R_t \le r\} \right) \\
    &= \sum_{t_0=1}^{T_0} \mathbb{P}(T = t_0)\,
         \sum_{t=1}^{t_0} \mathbb{E}\!\left(\mathbf{1}\{R_t \le r\} \,\big|\, T = t_0 \right) \\
    &= \sum_{t_0=1}^{T_0} \mathbb{P}(T = t_0)
        \sum_{t=1}^{t_0} \mathbb{P}\!\left( R_t \le r \right) \\
\implies \mathcal{F}'(r)
    &= \sum_{t_0=1}^{T_0} \mathbb{P}(T = t_0)
      \sum_{t=1}^{t_0}f_t(r).
\end{align*}
\begin{proof}[Proof of Theorem~\ref{thm1}] 
For $a \in [0,1]$, let 
\begin{align}\label{eq:rN1}
   \mathcal{R}_N(a) = \sup\{r \in \left[0, 1\right] : \mathbb{G}_{N}(r)-a \mathbb{F}_{N}(r) = \min_{s \in \left[0, 1\right]} (\mathbb{G}_{N}(s)-a \mathbb{F}_{N}(s))\}.
\end{align}
By the change of variable $r = r_0 + \delta N^{-1/3}$, $\delta \in \mathbb{R}$, we get
\begin{align}\label{eq:rN2}
&\mathcal{R}_N(h(r_0) + N^{-1/3} x) - r_0 \nonumber
\\ &\qquad =
N^{-1/3} \underset{\delta}{\operatorname{arg}\operatorname{min}}
\left\{
\mathbb{G}_{N}\left(r_0 + \delta N^{-1/3}\right)
-
\left(h(r_0) + N^{-1/3} x\right)\, \mathbb{F}_{N}\left(r_0 + \delta N^{-1/3}\right)
\right\}.
\end{align}
Then using Lemma 3.2 of \cite{MR3445293} we get, 
\begin{align}\label{eq0}
\mathbb{P}&\left(N^{1/3}\left(\hat{h}_{N}\left(r_{0}\right)-h\left(r_{0}\right)\right) \leq x\right) \nonumber
\\ &=\mathbb{P}\left(\hat{h}_{N}\left(r_{0}\right)\leq h(r_0) + N^{-1/3} x\right) \nonumber
\\ &= \mathbb{P}\left(\mathcal{R}_N(h(r_0) + N^{-1/3} x) \geq r_0\right) \nonumber
\\ & =\mathbb{P}\left(\underset{\delta}{\operatorname{arg}\operatorname{min}} \left\{\mathbb{G}_{N}\left(r_{0}+N^{-1/3} 
\delta\right)-\left(h\left(r_{0}\right)+N^{-1/3}x\right)\mathbb{F}_{N}\left(r_{0}+N^{-1/3}\delta\right)\right\} \geq 0\right).
\end{align}
Observe that
\begin{align*}
& \underset{\delta}{\operatorname{arg}\operatorname{min}} \Bigg\{\mathbb{G}_{N}\left(r_{0}+N^{-1/3} 
\delta\right)-\left(h\left(r_{0}\right)+N^{-1/3}x\right)\mathbb{F}_{N}\left(r_{0}+N^{-1/3}\delta\right)\Bigg\}
\\ &= \underset{\delta}{\operatorname{arg}\operatorname{min}} \big\{\sqrt{N}(\mathbb{P}_N-P){\alpha}_{N,\delta} -  N^{-1/3} x\sqrt{N}(\mathbb{P}_N-P){\beta}_{N,\delta}\\
&\qquad \qquad \quad + \sqrt{N}P\alpha_{N,\delta}
- N^{1/6} xP{\beta}_{N,\delta}
\big\}, \quad \text{ where}\\
\alpha_{N,\delta}&(T, (U_{t+1}, R_t)_{t=1}^{T}) = N^{1/6}\sum_{t=1}^{T}(U_{t+1}-h(r_0))\left(\mathbf{1}\left\{R_{t}\le r_0+N^{-1/3}\delta\right\}-\mathbf{1}\left\{R_{t}\le r_0\right\}\right), \\
\beta_{N,\delta}&(T, (U_{t+1}, R_t)_{t=1}^{T}) = N^{1/6}\sum_{t=1}^{T}\left(\mathbf{1}\left\{R_{t}\le r_0+N^{-1/3}\delta\right\}-\mathbf{1}\left\{R_{t}\le r_0\right\}\right), 
\end{align*}
First, for \(\delta \in [0,K]\), $K > 0$ (the case $\delta \in [-K,0]$ follows similarly), the Mean Value Theorem implies that, uniformly in $\delta \in [0,K]$,
\begin{align}\label{eq:clt1}
\sqrt{N}P\alpha_{N,\delta} &= N^{2/3} P\left[\sum_{t=1}^T (U_{t+1}-h(r_0))\left(\mathbf{1}\left\{R_{t}\le r_0+N^{-1/3}\delta\right\}
- \mathbf{1}\left\{R_{t}\le r_0\right\}\right)\right] \nonumber
\\ &= N^{2/3} \sum_{t_0=1}^{T_0} \mathbb{P}(T=t_0) \sum_{t=1}^{t_0}\int_{r_0}^{r_0+N^{-1/3}\delta}(h(r)-h(r_0))f_t(r)dr \nonumber
\\ &= N^{2/3} \sum_{t_0=1}^{T_0} \mathbb{P}(T=t_0) \sum_{t=1}^{t_0}\int_{r_0}^{r_0+N^{-1/3}\delta}h'(\xi(r))(r-r_0)f_t(r)dr \nonumber
\\ &\rightarrow\frac{1}{2}h'(r_0)\delta^2\sum_{t_0=1}^{T_0} \mathbb{P}(T=t_0) \sum_{t=1}^{t_0}f_t(r_0), \quad N \to \infty, \nonumber
\\ &= \frac{1}{2} h^{\prime}\left(r_{0}\right) \mathcal{F}^{\prime}\left(r_{0}\right) \delta^{2}.
\end{align}
Similarly, uniformly in $\delta \in [-K,K]$, $K > 0$, 
\begin{align}\label{eq:clt2}
N^{1/6} xP{\beta}_{N,\delta} \to \mathcal{F}^{\prime}\left(r_{0}\right)x\delta, \quad N \to \infty.
\end{align}
Next, a routine but tedious verification (see, for example, the proof of Theorem 4.3 of \cite{WellnerZhang1998}) shows that the function class  $\mathcal{A}_{N,K} := \{\alpha_{N,\delta} : \delta \in[-K,K]\}$, $K > 0$ satisfies the assumptions of Theorem 2.11.23 of \cite{MR1385671}. In particular, the pointwise limit of the covariance \(P \alpha_{N,\delta_1} \alpha_{N,\delta_2} - P \alpha_{N,\delta_1} P \alpha_{N,\delta_2}\) can be obtained as follows. First, from \eqref{eq:clt1}, for any \(\delta \in [-K,K]\), 
$P \alpha_{N,\delta} \to 0$. Second, for any \(\delta_1,\delta_2 \in [0,K]\), 
\begin{align*}
\alpha_{N,\delta_1}\alpha_{N,\delta_2}
&= N^{1/3}\Bigg[\sum_{t=1}^{T}(U_{t+1}-h(r_0))^2\mathbf{1}\left\{r_0<R_{t}\le r_0+N^{-1/3}(\delta_1\wedge \delta_2)\right\}
\\ &\qquad \qquad + \sum_{t=1}^{T}\sum_{s\ne t}(U_{t+1}-h(r_0))(U_{s+1}-h(r_0))\\ &\hspace{3cm}\times\mathbf{1}\left\{r_0<R_{t}\le r_0+N^{-1/3}\delta_1\right\}
\mathbf{1}\left\{r_0<R_{s}\le r_0+N^{-1/3}\delta_2\right\}\Bigg].
\end{align*}
It follows that
\begin{align*}
&N^{1/3}\Bigg|P\Bigg[\sum_{t=1}^{T}\sum_{s\ne t}(U_{t+1}-h(r_0))(U_{s+1}-h(r_0))\\ &\hspace{3cm}\times\mathbf{1}\left\{r_0<R_{t}\le r_0+N^{-1/3}\delta_1\right\}
\mathbf{1}\left\{r_0<R_{s}\le r_0+N^{-1/3}\delta_2\right\}\Bigg]\Bigg|
\\ &\le C_1 N^{1/3}\sum_{t_0=1}^{T_0}\mathbb{P}(T=t_0)
\sum_{t=1}^{t_0}\sum_{s\neq t}\mathbb{P}\left(r_0<R_{t}\le r_0+N^{-1/3}\delta_1, r_0<R_{s}\le r_0+N^{-1/3}\delta_2\right)
\\ &\le C_2 N^{-1/3}\rightarrow 0, \hspace{7cm} \text{ and }
\\[.3cm]
&N^{1/3}P\left[\sum_{t=1}^{T}(U_{t+1}-h(r_0))^2\mathbf{1}\left\{r_0<R_{t}\le r_0+N^{-1/3}(\delta_1\wedge \delta_2)\right\}\right]
\\ &= N^{1/3}\mathbb{E}\left[\sum_{t=1}^{T}\left(h(R_t)(1-h(R_t))+(h(R_t)-h(r_0))^2\right)\mathbf{1}\left\{r_0<R_{t}\le r_0+N^{-1/3}(\delta_1\wedge \delta_2)\right\}\right]
\\ &= N^{1/3}\left[\sum_{t_0=1}^{T_0}\mathbb{P}(T=t_0)\sum_{t=1}^{t_0}\int_{r_0}^{r_0+N^{-1/3}(\delta_1\wedge \delta_2)}h(r)(1-h(r))f_t(r)dr\right.
\\ &\qquad \qquad\left.+\sum_{t_0=1}^{T_0}\mathbb{P}(T=t_0)\sum_{t=1}^{t_0}\int_{r_0}^{r_0+N^{-1/3}(\delta_1\wedge \delta_2)}(h(r)-h(r_0))^2f_t(r)dr\right]
\\ &\rightarrow(\delta_1\wedge \delta_2)h(r_0)(1-h(r_0))\sum_{t_0=1}^{T_0}\mathbb{P}(T=t_0)\sum_{t=1}^{t_0}f_t(r_0)=(\delta_1\wedge \delta_2)h(r_0)(1-h(r_0))\mathcal{F}'(r_0).
\end{align*}
Hence for any \(\delta_1,\delta_2 \in [0,K]\),
\[
P \alpha_{N,\delta_1} \alpha_{N,\delta_2} - P \alpha_{N,\delta_1} P \alpha_{N,\delta_2}
\rightarrow
(\delta_1\wedge \delta_2)h(r_0)(1-h(r_0))\mathcal{F}'(r_0).
\]
Similarly, for any \(\delta_1,\delta_2 \in [-K,0]\),
\[
P \alpha_{N,\delta_1} \alpha_{N,\delta_2} - P \alpha_{N,\delta_1} P \alpha_{N,\delta_2}
\rightarrow
-(\delta_1\vee \delta_2)h(r_0)(1-h(r_0))\mathcal{F}'(r_0),
\]
and for any \(\delta_1,\delta_2 \in [-K,K]\) with $\delta_1\times\delta_2 < 0$, as $\alpha_{N,\delta_1}\alpha_{N,\delta_2} \equiv 0$, 
\[
P \alpha_{N,\delta_1} \alpha_{N,\delta_2} - P \alpha_{N,\delta_1} P \alpha_{N,\delta_2}
\rightarrow 0.
\]
Consequently, Theorem 2.11.23 of \cite{MR1385671} implies that
\begin{align}\label{eq:clt3}
\sqrt{N}(\mathbb{P}_N-P){\alpha}_{N,\delta} \xrightarrow{d} \sqrt{h(r_0)(1-h(r_0))\mathcal{F}^{\prime}\left(r_{0}\right)} \mathbb{B}(\delta)
\quad \text{ in } \quad \mathcal{L}_\infty. 
\end{align}
By a similar argument,
\begin{align}\label{eq:clt4}
    N^{-1/3}\sqrt{N}(\mathbb{P}_N-P){\beta}_{N,\delta} \xrightarrow{d} 0
\quad \text{ in } \quad \mathcal{L}_\infty. 
\end{align}
Combining \eqref{eq:clt1} - \eqref{eq:clt4} and Slutsky's Theorem, we get
\begin{align}\label{eq:cltm}
& N^{2/3}\left[\mathbb{G}_{N}\left(r_{0}+N^{-1/3} \delta\right)-\left(h\left(r_{0}\right)+N^{-1 / 3} x\right) \mathbb{F}_{N}\left(r_{0}+N^{-1/3}\delta\right)\right] \nonumber\\
& \quad \stackrel{d}{\rightarrow} \sqrt{h(r_0)(1-h(r_0))\mathcal{F}^{\prime}\left(r_{0}\right)} \mathbb{B}(\delta)+\frac{1}{2} h^{\prime}\left(r_{0}\right) \mathcal{F}^{\prime}\left(r_{0}\right) \delta^{2}-\mathcal{F}^{\prime}\left(r_{0}\right)x\delta \quad \text { in } \quad \mathcal{L}_\infty.
\end{align}
Using Theorems 2.14.2 and 3.2.5 of \cite{MR1385671}, it can be verified that (see, for example, the proof of Theorem 4.3 of \cite{WellnerZhang1998}),
\[\underset{\delta}{\operatorname{arg}\operatorname{min}} \Bigg\{\mathbb{G}_{N}\left(r_{0}+N^{-1/3} 
\delta\right)-\left(h\left(r_{0}\right)+N^{-1/3}x\right)\mathbb{F}_{N}\left(r_{0}+N^{-1/3}\delta\right)\Bigg\} = O_p(1).\]
Then Theorem 3.2.2 of \cite{MR1385671} and \eqref{eq:cltm} imply that
\begin{align}\label{eq:cltm1}
&\underset{\delta}{\operatorname{arg}\operatorname{min}} \Bigg\{\mathbb{G}_{N}\left(r_{0}+N^{-1/3} 
\delta\right)-\left(h\left(r_{0}\right)+N^{-1/3}x\right)\mathbb{F}_{N}\left(r_{0}+N^{-1/3}\delta\right)\Bigg\} \nonumber
\\ &\stackrel{d}{\to} \underset{\delta}{\arg \min }\left\{\sqrt{h(r_0)(1-h(r_0))\mathcal{F}^{\prime}\left(r_{0}\right)} \mathbb{B}(\delta)+\frac{1}{2} h^{\prime}\left(r_{0}\right) \mathcal{F}^{\prime}\left(r_{0}\right) \delta^{2}-\mathcal{F}^{\prime}\left(r_{0}\right)x\delta\right\}. 
\end{align}
An application of the scaling property of Brownian motion yields (see, for example, Exercise 3.27 of \cite{MR3445293}), 
\begin{align}\label{eq:cltm2}
&\underset{\delta}{\arg \min }\left\{\sqrt{h(r_0)(1-h(r_0))\mathcal{F}^{\prime}\left(r_{0}\right)} \mathbb{B}(\delta)+\frac{1}{2} h^{\prime}\left(r_{0}\right) \mathcal{F}^{\prime}\left(r_{0}\right) \delta^{2}-\mathcal{F}^{\prime}\left(r_{0}\right)x\delta\right\} \nonumber
\\ & \quad\stackrel{d}{=}\left[\frac{4 h(r_0)(1-h(r_0))}{\left(h^{\prime }\left(r_{0}\right)\right)^2 \mathcal{F}^{\prime}\left(r_{0}\right)}\right]^{1 / 3} \underset{\delta}{\arg \min }\left\{\mathbb{B}(\delta)+\delta^{2}\right\}+\frac{x}{h^{\prime}\left(r_{0}\right)}.
\end{align}
Combining \eqref{eq0}, \eqref{eq:cltm1} and \eqref{eq:cltm2}, we get
$$
\begin{aligned}
& \mathbb{P}\left(N^{1 / 3}\left(\hat{h}_{N}\left(r_{0}\right)-h\left(r_{0}\right)\right) \leq x\right) \\
& \quad \rightarrow \mathbb{P}\left(\left[\frac{4 h(r_0)(1-h(r_0))}{\left(h^{\prime }\left(r_{0}\right)\right)^2 \mathcal{F}^{\prime}\left(r_{0}\right)}\right]^{1 / 3} \underset{\delta}{\arg \min }\left\{\mathbb{B}(\delta)+\delta^{2}\right\} \geq-\frac{x}{h^{\prime}\left(r_{0}\right)}\right) \\
& \quad=\mathbb{P}\left(\left[\frac{4h(r_0)(1-h(r_0)) h^{\prime}\left(r_{0}\right)}{\mathcal{F}^{\prime}\left(r_{0}\right)}\right]^{1 / 3} \underset{\delta}{\arg \max }\left\{\mathbb{B}(\delta)-\delta^{2}\right\} \leq x\right),
\end{aligned}
$$
completing the proof.
\end{proof}

\begin{proof}[Proof of Theorem~\ref{thm2}] 
Fix \((\delta_{1},\ldots,\delta_{k}) \in \mathbb{R}^k\). 
It is enough to show that 
\begin{align}\label{eq:fdd}
\{\Omega_{N}(\delta_{i}), \Omega^{(0)}_{N}(\delta_{i})\}_{i=1}^{k} &\;\;\stackrel{d}{\to}\;\; 
\{ g_{\alpha, \beta}(\delta_i), g^{(0)}_{\alpha, \beta}(\delta_i)\}_{i=1}^{k}.
\end{align}
This is because the processes $\Omega_{N}$, $\Omega^{(0)}_{N}$ are monotone, and so \eqref{eq:fdd} and Corollary 2 of \cite{MR1311975} imply 
the convergence in \eqref{eq:fdd} also hold in the space \(\mathcal{L}_{2} \times \mathcal{L}_{2}\). The proof of \eqref{eq:fdd} follows along the same lines as that of Theorem~\ref{thm1}; we outline the main steps below.

Define the processes \(\mathbb{M}_{N}\) and $\mathbb{U}_{N}$ as
\begin{align*}
\mathbb{M}_{N}(\delta) &= N^{2/3} \, \mathcal{F}'(r_{0})^{-1} 
\Big( \mathbb{G}_{N}(r_{0} + N^{-1/3}\delta) - \mathbb{G}_{N}(r_{0}) 
- h(r_0) \big( \mathbb{F}_N(r_{0} + N^{-1/3}\delta) - \mathbb{F}_N(r_{0}) \big) \Big), \\
\mathbb{U}_{N}(\delta) &= N^{1/3} \, \mathcal{F}'(r_{0})^{-1} 
\big( \mathbb{F}_N(r_{0} + N^{-1/3}\delta) - \mathbb{F}_N(r_{0}) \big).
\end{align*}
For $a \in [0,1]$ and $A \in \mathbb{R}$, let 
\begin{align*}
   \mathcal{R}^{+}_N(a) &= \inf\{r \in \left[r_0, 1\right] : \mathbb{G}_{N}(r)-a \mathbb{F}_{N}(r) = \min_{s \in \left[r_0, 1\right]} (\mathbb{G}_{N}(s)-a \mathbb{F}_{N}(s))\},\\
   \mathcal{R}^{-}_N(a) &= \sup\{r \in \left[0, r_0\right] : \mathbb{G}_{N}(r)-a \mathbb{F}_{N}(r) = \min_{s \in \left[0, r_0\right]} (\mathbb{G}_{N}(s)-a \mathbb{F}_{N}(s))\},\\
   \widetilde{\mathcal{R}}_N(A) &= \sup\{\delta \in \mathbb{R} : \mathbb{M}_{N}(\delta)-A\mathbb{U}_{N}(\delta) = \min_{h \in \mathbb{R}} (\mathbb{M}_{N}(h)-A \mathbb{U}_{N}(h))\},
   \\ \widetilde{\mathcal{R}}^{+}_N(A) &= \inf\{\delta \in \mathbb{R}^{+} : \mathbb{M}_{N}(\delta)-A\mathbb{U}_{N}(\delta) = \min_{h \in \mathbb{R}^{+}} (\mathbb{M}_{N}(h)-A \mathbb{U}_{N}(h))\},
   \\ \widetilde{\mathcal{R}}^{-}_N(A) &= \sup\{\delta \in \mathbb{R}^{-} : \mathbb{M}_{N}(\delta)-A\mathbb{U}_{N}(\delta) = \min_{h \in \mathbb{R}^{-}} (\mathbb{M}_{N}(h)-A \mathbb{U}_{N}(h))\}.
\end{align*}
From \eqref{eq:rN1} and \eqref{eq:rN2}, centering the
processes $\mathbb{F}_N$ (respectively, $\mathbb{G}_N$) around $\mathbb{F}_N(r_0)$ (respectively, $\mathbb{G}_N(r_0)$), multiplying them a factor of $N^{2/3}$, and applying the change of variable $r = r_0 + \delta N^{-1/3}$, $\delta \in \mathbb{R}$, we get
\begin{align*}
\mathcal{R}_N(h(r_0) + N^{-1/3} A) = r_0 + N^{-1/3}\widetilde{\mathcal{R}}_N(A), \\
\mathcal{R}^{\pm}_N(h(r_0) + N^{-1/3} A) = r_0 + N^{-1/3}\widetilde{\mathcal{R}}^{\pm}_N(A).
\end{align*}
Note that, as shown in \eqref{eq0}, for $\delta, x \in \mathbb{R}$, 
\begin{align}\label{eq:rN4}
    \Omega_{N}(\delta) \leq x \iff \mathcal{R}_N\left(h(r_0) + N^{-1/3}x\right) \geq r_0 + N^{-1/3}\delta \iff \widetilde{\mathcal{R}}_N(x)\geq \delta.
\end{align}
Similarly, for $\delta, x \in \mathbb{R}^{\pm}$,
\begin{align}\label{eq:rN5}
    \Omega_{N}^{(0)}(\delta) \gtrless x \iff \mathcal{R}^{\pm}_N\left(h(r_0) + N^{-1/3}x\right) \lessgtr r_0 + N^{-1/3}\delta \iff \widetilde{\mathcal{R}}^{\pm}_N(x)\lessgtr \delta.
\end{align}
It follows, from \eqref{eq:bm}, \eqref{eq:clt2}, \eqref{eq:cltm} and Slutsky's Theorem,  that, 
\begin{align}\label{eq:c1}
    \left(\mathbb{M}_{N}(\delta), \mathbb{U}_{N}(\delta)\right) \stackrel{d}{\to} (\widetilde{\mathbb{B}}_{\alpha, \beta}(\delta), \delta) \quad \text { in } \quad \mathcal{L}_\infty \times \mathcal{L}_\infty.
\end{align}  
From Lemma 3.2 of \cite{MR3445293}, it also follows that, for $\delta, x \in \mathbb{R}$, almost surely
\begin{align}\label{eq:c01}
   \underset{h}{\operatorname{arg}\operatorname{min}} \left(\widetilde{\mathbb{B}}_{\alpha, \beta}(h) - xh\right) \geq \delta \iff g_{\alpha, \beta}(\delta) \leq x, 
\end{align}
and for $\delta, x \in \mathbb{R}^{\pm}$, almost surely
\begin{align}\label{eq:c02}
   \underset{h \in \mathbb{R}^{\pm}}{\operatorname{arg}\operatorname{min}} \left(\widetilde{\mathbb{B}}_{\alpha, \beta}(h) - xh\right) \geq \delta \iff g^{\pm}_{\alpha, \beta}(\delta) \leq x, 
\end{align}
Then Theorem 3.2.2 of \cite{MR1385671} (as argued in the proof of Theorem~\ref{thm:lrt}), along with \eqref{eq:rN4} - \eqref{eq:c02}, imply \eqref{eq:fdd}.
\end{proof}

The following Proposition~\ref{p1} and Proposition~\ref{lem:dist} are standard results (see, for example, \cite{banerjee2001likelihood}) which we need for proving Theorem~\ref{thm:lrt}. Denote by 
$E_{a,b}$ the set on which the two processes \(g_{a,b}\) and \(g^{0}_{a,b}\) differ. 

\begin{proposition}\label{p1}
Fix $a,b > 0$. For any $\epsilon > 0$, there exists $L_\epsilon  = L_\epsilon(a,b) > 0$ such that $\mathbb{P}\left(E_{a,b} \subseteq [-L_\epsilon,L_\epsilon]\right) \geq 1-\epsilon$.
\end{proposition}

By an application of Brownian scaling, the processes 
\(g_{a,b}\) and \(g^{0}_{a,b}\) can be shown to be related in distribution to what are called the canonical processes \(g_{1,1}\) and \(g^{0}_{1,1}\).

\begin{proposition}\label{lem:dist}
For any \(a,b > 0\), we have, in the space $\mathcal{L}_2 \times \mathcal{L}_2$,
\begin{align*}
&\bigl(g_{a,b}(x),\, g^0_{a,b}(x),\, E_{a,b}\bigr) \stackrel{d}{=} a(b/a)^{1/3}
\Bigl(g_{1,1}\bigl((b/a)^{2/3} x\bigr),\,g^0_{1,1}\bigl((b/a)^{2/3} x\bigr),\,(a/b)^{2/3} E_{1,1}\Bigr),\\
&\hspace{3.5cm}\frac{1}{a^2} \int \Big\{ (g_{a,b}(x))^2 - (g^{0}_{a,b}(x))^2 \Big\}\, dx 
\;\;\stackrel{d}{=}\;\; \mathbb{D}.
\end{align*}
\end{proposition}

Further denote by $E_{N}$ the set on which $\hat{h}_{N}$ and $\hat{h}_{N}^{(0)}$ differ and let $\widetilde{E}_{N}$ to be the set $N^{1 / 3}\left(E_{N}-r_{0}\right)$. We then have the following result similar to Proposition~\ref{p1}.

\begin{lemma}\label{lem1}
    For any $\epsilon > 0$, there exists $L_\epsilon > 0$ such that with probability at least $1-\epsilon$, $\widetilde{E}_{N} \subseteq [-L_\epsilon,L_\epsilon]$ for all sufficiently large $N$.
\end{lemma}

\begin{proof}
By definition, $\widetilde{E}_{N}$ is either the null set, or an interval $[X_N, Y_N]$ containing the origin, and in the first case the proof is immediate. For the latter case, fix $\epsilon > 0$. 
For any $M > 0$, $Y_N > M$ implies that $r_0 + N^{-1/3}M \in E_N$. Note that $r \in E_N$ implies either i) $\hat{h}_N(r) = \hat{h}_N(r_0)$ or ii) either $\hat{h}^{(0)}_N(r) = h_0 < \hat{h}_N(r) < \hat{h}_N(r_0)$ or $\hat{h}_N(r_0) < \hat{h}_N(r) < h_0 = \hat{h}^{(0)}_N(r)$. So,
\begin{equation*}
\{Y_N > M\} \subset  \{\Omega_N(M) = \Omega_N(0)\} \cup  \{\Omega^{(0)}_N(M) = 0\}.
\end{equation*}
By Theorem \ref{thm2}, for all sufficiently large $M$, we have
\begin{align*}&\mathbb{P}(\Omega_N(M) = \Omega_N(0)) \to \mathbb{P}(g_{\alpha,\beta}(M) = g_{\alpha,\beta}(0)) < \epsilon/4, \\
&\mathbb{P}(\Omega^{(0)}_N(M) = 0) \to \mathbb{P}(g^{(0)}_{\alpha,\beta}(M) = 0) = \mathbb{P}(g^+_{\alpha,\beta}(M) \leq 0)  < \epsilon/4.\end{align*}
Thus there exists $L_\epsilon > 0$ such that $\mathbb{P}(Y_N > L_{\epsilon}) < \epsilon/2$ for all sufficiently large $N$. In a similar way, we can show (by choosing $L_\epsilon$ so that it works in both cases) $\mathbb{P}(X_N < -L_{\epsilon}) < \epsilon/2$ for all sufficiently large $N$. Thus, $\mathbb{P}(\widetilde{E}_{N} \subseteq [-L_\epsilon,L_\epsilon]) > 1-\epsilon$ for all sufficiently large $N$.
\end{proof}

\begin{proof}[Proof of Theorem~\ref{thm:lrt}]
Let $\theta(x) = \log \left(1+e^x\right)$, $\psi(r) = \log \left(\frac{h(r)}{1-h(r)}\right)$, $\psi_0 = \psi(r_0)$,
\[\hat{\psi}_{N}(r) = \log\left(\frac{\hat{h}_N(r)}{1-\hat{h}_N(r)}\right),  \quad \hat{\psi}^{(0)}_{N}(r) = \log\left(\frac{\hat{h}^{(0)}_N(r)}{1-\hat{h}^{(0)}_N(r)}\right),\] and $\hat{\psi}_{s} = \hat{\psi}_{N}\left(R_{(s)}\right)$, $\hat{\psi}^{(0)}_{s} = \hat{\psi}_{N}^{(0)}\left(R_{(s)}\right)$. Applying Delta method on \eqref{eq:fdd}, 
we get
\begin{align}\label{eqd}
    N^{1/3}&\left(\hat{\psi}_N\left(r_0 + N^{-1/3}x\right) - \psi(r_0), \hat{\psi}^{(0)}_N\left(r_0 + N^{-1/3}x\right) - \psi(r_0)\right) \stackrel{\operatorname{d}}{\to} \left(g_{\eta,\xi}, g^{(0)}_{\eta,\xi}\right),\\
\eta^{2} &= \left(\frac{1}{h\left(r_{0}\right)\left(1-h\left(r_{0}\right)\right)}\right)^2\alpha^2 = \frac{1}{h\left(r_{0}\right)\left(1-h\left(r_{0}\right)\right)\mathcal{F}^\prime\left(r_{0}\right)}, \nonumber \\ \xi &= \left(\frac{1}{h\left(r_{0}\right)\left(1-h\left(r_{0}\right)\right)}\right)\beta = \frac{h'(r_0)}{2h\left(r_{0}\right)\left(1-h\left(r_{0}\right)\right)}. \nonumber\end{align} 
Note that Delta method on \eqref{eq:fdd} yields only the convergence of the finite dimensional distributions. As the processes $\hat{\psi}_N$, $ \hat{\psi}^{(0)}_N$ are monotone, Corollary 2 of \cite{MR1311975} implies
the convergence in \eqref{eqd} is 
in the space $\mathcal{L}_2 \times \mathcal{L}_2$ also. The likelihood ratio statistic is then given by
$$
\begin{aligned}
2 \log \lambda_{N}(h_0) 
& =2 \sum_{s=1}^{S} {F}_{s} \overline{U}_{s}\left(\hat{\psi}_{s}-\hat{\psi}^{(0)}_{s}\right)-2 \sum_{s=1}^{S} {F}_{s}\left(\theta\left(\hat{\psi}_{s}\right)-\theta\left(\hat{\psi}_{s}^{(0)}\right)\right).
\end{aligned}
$$
By denoting the set of indices $s \in \{1, \ldots, S\}$ by $I_N$ for which $\hat{\psi}_{s}$ differs from $\hat{\psi}_{s}^{(0)}$ and using Taylor expansion of $\theta\left(\hat{\psi}_{s}\right)$ and $\theta\left(\hat{\psi}_{s}^{(0)}\right)$ around $\psi_0
$, we obtain
\begin{align}\label{eq:ll1}
2 \log \lambda_{N}(h_0) &= 2 \sum_{s=1}^{S} {F}_{s} \overline{U}_{s}\left(\hat{\psi}_{s}-\hat{\psi}^{(0)}_{s}\right) - 2 \theta'\left(\psi_0\right)\sum_{s \in I_N} {F}_s\left(\left(\hat{\psi}_{s}-\psi_0\right)-\left(\hat{\psi}_{s}^{(0)}-\psi_0\right)\right) \nonumber\\ &\qquad -\theta''\left(\psi_0\right)\sum_{s \in I_N} {F}_s\left(\left(\hat{\psi}_{s}-\psi_0\right)^2-\left(\hat{\psi}_{s}^{(0)}-\psi_0\right)^2\right) +  A^{(1)}_N\nonumber\\
&= 2 \sum_{s \in I_N} {F}_s\left(\hat{\psi}_{s}-{\psi}_{0}\right)\left(\overline{U}_{s} - \theta'\left(\psi_0\right)\right) - 2 \sum_{s \in I_N} {F}_s\left(\hat{\psi}_{s}^{(0)}-{\psi}_{0}\right)\left(\overline{U}_{s} - \theta'\left(\psi_0\right)\right) \nonumber\\ & \qquad -\theta''\left(\psi_0\right)\sum_{s \in I_N} {F}_s\left(\left(\hat{\psi}_{s}-\psi_0\right)^2-\left(\hat{\psi}_{s}^{(0)}-\psi_0\right)^2\right) + A^{(1)}_N,
\end{align}
where for some $\psi^{(1)}_{N,s}$ (lying between $\hat{\psi}_{s}$ and $\psi_0$) and $\psi^{(2)}_{N,s}$ (lying between $\hat{\psi}_{s}^{(0)}$ and $\psi_0$),
\[
   A^{(1)}_N = -\frac{1}{3}\sum_{s \in I_N} {F}_s\left(\theta'''\left(\psi^{(1)}_{N,s}\right)\left(\hat{\psi}_{s}-\psi_0\right)^3-\theta'''\left(\psi^{(2)}_{N,s}\right)\left(\hat{\psi}_{s}^{(0)}-\psi_0\right)^3\right).
\] 
Divide $I_{N}$ into consecutive sets of indices $S_{1}, S_{2}, \ldots, S_{k}$ such that, 
$\hat{\psi}_{i}$ is constant ( = $c_j$, say) for all $i \in S_{j}$, $1 \leq j \leq k$. On $S_{j}$, from \eqref{eq:hath}, it holds that $\theta'\left(c_{j}\right)
= \frac{\sum_{i \in S_{j}} {F}_{i} \overline{U}_{i}}{\sum_{i \in S_{j}} {F}_{i}}$, and so
\begin{align}\label{eq:ll2}
&\sum_{s \in I_N} {F}_s\left(\hat{\psi}_{s}-{\psi}_{0}\right)\left(\overline{U}_{s} - \theta'\left(\psi_0\right)\right) \nonumber
\\& = \sum_{j = 1}^k\left(c_{j}-\psi_0\right)\left(\sum_{i \in S_{j}} {{F}_{i} \overline{U}_{i}}-\theta'\left(\psi_0\right) \sum_{i \in S_{j}} {{F}_{i}}\right) \nonumber
\\ & = \sum_{j = 1}^k\left(c_{j}-\psi_0\right)\left(\theta'\left(c_j\right) - \theta'\left(\psi_{0}\right)\right) \sum_{i \in S_{j}} {{F}_{i}} \nonumber
\\ &=\sum_{s \in I_N} {F}_s\left(\hat{\psi}_{s}-{\psi}_{0}\right)\left(\theta'\left(\hat{\psi}_{s}\right) - \theta'\left(\psi_{0}\right)\right) \nonumber
\\ &=\theta''\left(\psi_0\right)\sum_{s \in I_N} {F}_s\left(\hat{\psi}_{s}-{\psi}_{0}\right)^2 + A^{(2)}_N,
\end{align}
where for some $\psi^{(3)}_{N,s}$ (lying between $\hat{\psi}_{s}$ and $\psi_0$),
\[
   A^{(2)}_N = \frac{1}{2}\sum_{s \in I_N} {F}_s\theta'''\left(\psi^{(3)}_{N,s}\right)\left(\hat{\psi}_{s}-\psi_0\right)^3.
\]
Similarly, we have
\begin{align}\label{eq:ll3}
\sum_{s \in I_N} {F}_s\left(\hat{\psi}^{(0)}_{s}-{\psi}_{0}\right)\left(\overline{U}_{s} - \theta'\left(\psi_0\right)\right)
=\theta''\left(\psi_0\right)\sum_{s \in I_N} {F}_s\left(\hat{\psi}^{(0)}_{s}-{\psi}_{0}\right)^2 + A^{(3)}_N,
\end{align}
where for some $\psi^{(4)}_{N,s}$ (lying between $\hat{\psi}^{(0)}_{s}$ and $\psi_0$),
\[
   A^{(3)}_N = \frac{1}{2}\sum_{s \in I_N} {F}_s\theta'''\left(\psi^{(4)}_{N,s}\right)\left(\hat{\psi}^{(0)}_{s}-\psi_0\right)^3.
\]
Combining \eqref{eq:ll1}, \eqref{eq:ll2} and \eqref{eq:ll3}, we get
\begin{align}\label{eq:ll4}
    2 \log \lambda_{N}(h_0) &= 
\theta''\left(\psi_0\right)\sum_{s \in I_N} {F}_s\left(\left(\hat{\psi}_{s}-\psi_0\right)^2-\left(\hat{\psi}_{s}^{(0)}-\psi_0\right)^2\right) + A^{(1)}_N +2(A^{(2)}_N - A^{(3)}_N).
\end{align}
Denote by $E_{N}$ the set on which $\hat{\psi}_{N}$ and $\hat{\psi}_{N}^{(0)}$ differ. We then have 
\begin{align}\label{eq:ll5}
& \theta''\left(\psi_0\right)\sum_{s \in I_N} {F}_s\left(\left(\hat{\psi}_{s}-\psi_0\right)^2-\left(\hat{\psi}_{s}^{(0)}-\psi_0\right)^2\right) = N^{1/3}(\mathbb{P}_N-P)\gamma_{N} + N^{1/3}P\gamma_{N},
\quad \text{ where } \nonumber
    \\[.15cm] &\gamma_{N}(T, (R_t)_{t=1}^{T}) = \theta''\left(\psi_0\right)\sum_{t=1}^{T_0}\left\{\left(N^{1/3}\left(\hat{\psi}_{N}(R_t)-\psi_0\right)\right)^{2} \right.\\ &\hspace{5.15cm}\left. -\left(N^{1/3}\left(\hat{\psi}_{N}^{(0)}(R_t)-\psi_0\right)\right)^{2}\right\}\mathbf{1}\left\{R_t \in E_{N}\right\}\mathbf{1}\left\{t \leq T\right\} \nonumber
\end{align}
Note that, due to \eqref{eqd}, the monotone functions $N^{1/3}\left(\hat{\psi}_N(r) - \psi_0\right)$ and $N^{1/3}\left(\hat{\psi}_N^{(0)}(r) - \psi_0\right)$ are, with arbitrarily high probability, bounded on compact sets of the form $[r_0 - N^{-1/3}M, r_0 + N^{-1/3}M]$ for all sufficiently large $N$. From Lemma~\ref{lem1}, it follows that, 
with arbitrarily high probability, $E_N$ is contained in intervals of the form $[r_0 - N^{-1/3}M, r_0 + N^{-1/3}M]$ for all sufficiently large $N$. Theorem 2.7.5 of \cite{MR1385671} and permanence of the Donsker property (see, for example, Section 2.10 of \cite{MR1385671}) thus imply that, with arbitrarily high probability, $\gamma_N$ is in a universally Donsker class of functions for all sufficiently large $N$. Consequently, 
\begin{align}\label{eq:ll6}
N^{1/3}(\mathbb{P}_N-P)\gamma_{N} = N^{-1/6}\sqrt{N}(\mathbb{P}_N-P)\gamma_{N} \stackrel{p}{\to} 0.
\end{align}
Further, as $\theta'''$ is bounded, a similar argument gives,  with arbitrarily high probability, that for all sufficiently large $N$ and for some constant $C_1, L > 0$ and ,
\begin{align*}
    \left|A^{(2)}_N\right| \leq C_1 (\mathbb{P}_N-P)\zeta_{N} + CP\zeta_{N},
\end{align*}
where $\zeta_{N}(T, (R_t)_{t=1}^{T}) = \sum_{t=1}^{T_0}\mathbf{1}\left\{R_t \in[r_0 - N^{-1/3}L, r_0 + N^{-1/3}L]\right\}\mathbf{1}\left\{t \leq T\right\}$. A similar argument as in the case of $\gamma_N$ shows that $(\mathbb{P}_N-P)\zeta_{N} \stackrel{p}{\to} 0$. Also, for some constant $C_2 > 0$,
\begin{align*}
P\zeta_{N} &= \mathbb{E}\left[\sum_{t=1}^{T}\mathbf{1}\left\{r_0 - N^{-1/3}L \leq R_{t} \le r_0+N^{-1/3}L\right\}\right] \nonumber
\\ &= \sum_{t_0=1}^{T_0}\mathbb{P}(T=t_0)\sum_{t=1}^{t_0}\int_{r_0 - N^{-1/3}L}^{r_0+N^{-1/3}L}f_t(r)dr \leq C_2N^{-1/3} \to 0.
\end{align*}
So, $A^{(2)}_N \stackrel{p}{\to} 0$ and similarly, $A^{(1)}_N, A^{(3)}_N\stackrel{p}{\to} 0$. These along with \eqref{eq:ll4}, \eqref{eq:ll5} and \eqref{eq:ll6} imply
\begin{align}\label{eq:ll7}
    2 \log \lambda_{N}(h_0) - N^{1/3}P\gamma_{N} \stackrel{p}{\to} 0.
\end{align}
Next let $\widetilde{E}_{N}$ to be the set $N^{1 / 3}\left(E_{N}-r_{0}\right)$. Then
\begin{align}\label{eq:ll8}
    N^{1/3}P\gamma_{N} &= N^{1/3}\theta''\left(\psi_0\right)\mathbb{E}\left[\sum_{t=1}^{T}\left\{\left(N^{1/3}\left(\hat{\psi}_{N}(R_t)-\psi_0\right)\right)^{2} \right.\right. \nonumber \\ &\hspace{4cm}\left.\left. -\left(N^{1/3}\left(\hat{\psi}_{N}^{(0)}(R_t)-\psi_0\right)\right)^{2}\right\}\mathbf{1}\left\{R_t \in E_{N}\right\}\right] \nonumber \\
    &= \theta''\left(\psi_0\right)\sum_{t_0=1}^{T_0}\mathbb{P}(T=t_0)\sum_{t=1}^{t_0}\int_{\widetilde{E}_{N}}\left\{\left(N^{1/3}\left(\hat{\psi}_{N}(r_0 + N^{-1/3}\delta)-\psi_0\right)\right)^{2} \right. \\ &\hspace{4.25cm}\left. -\left(N^{1/3}\left(\hat{\psi}_{N}^{(0)}(r_0 + N^{-1/3}\delta)-\psi_0\right)\right)^{2}\right\}f_t(r_0 + N^{-1/3}\delta)d\delta. \nonumber 
\end{align}
Recall that $E_{\eta,\xi}$ is the set on which $g_{\eta,\xi}$ and $g^{(0)}_{\eta,\xi}$ differ. From lemma~\ref{lem1} and \eqref{eqd}), we conclude that for each $\epsilon > 0$, there exists $L_\epsilon > 0$ such that with probability at least $1-\epsilon$, $\widetilde{E}_{N} \in [-L_\epsilon, L_\epsilon]$ for all sufficiently large $N$ and $E_{\eta,\xi} \in [-L_\epsilon, L_\epsilon]$ such that on $[-L_\epsilon, L_\epsilon]$ both $\left(N^{1/3}\left(\hat{\psi}_{N}^{(0)}(r_0 + N^{-1/3}\delta)-\psi_0\right)\right)$ and $\left(N^{1/3}\left(\hat{\psi}_{N}(r_0 + N^{-1/3}\delta)-\psi_0\right)\right)$ are uniformly bounded for all sufficiently large $N$. This, along with \eqref{eqd}, \eqref{eq:ll8} and continuity of $f_t$ imply that, for some $A^{(4)}_N \stackrel{p}{\to} 0$,
\begin{align}\label{eq:ll9}
    N^{1/3}P\gamma_{N} &= \frac{1}{\eta^2}\int_{\widetilde{E}_{N}}\left\{\left(N^{1/3}\left(\hat{\psi}_{N}(r_0 + N^{-1/3}\delta)-\psi_0\right)\right)^{2} \right. \\ &\hspace{2.25cm}\left. -\left(N^{1/3}\left(\hat{\psi}_{N}^{(0)}(r_0 + N^{-1/3}\delta)-\psi_0\right)\right)^{2}\right\}d\delta + A^{(4)}_N. \nonumber 
\end{align}
Observe that
\begin{align}\label{eq:ll10}
    \lim_{\epsilon \to 0} \limsup_{N\to\infty} &\hspace{.1cm}\mathbb{P}\left(\frac{1}{\eta^2}\int_{[-L_\epsilon, L_\epsilon]\setminus\widetilde{E}_{N}}\left\{\left(N^{1/3}\left(\hat{\psi}_{N}(r_0 + N^{-1/3}\delta)-\psi_0\right)\right)^{2} \right.\right. \nonumber \\ &\qquad\qquad\left.\left. -\left(N^{1/3}\left(\hat{\psi}_{N}^{(0)}(r_0 + N^{-1/3}\delta)-\psi_0\right)\right)^{2}\right\}d\delta \neq 0\right) = 0.
\end{align}
\begin{align}\label{eq:ll12}
    \lim_{\epsilon \to 0} \mathbb{P}\left(\frac{1}{\eta^2}\int_{[-L_\epsilon, L_\epsilon]\setminus{E}_{\eta,\xi}}\left\{\left(g_{\eta, b}(\delta)\right)^{2}-\left(g_{\eta, b}^{0}(\delta)\right)^{2}\right\} d \delta \neq 0\right) = 0.
\end{align}
Also \eqref{eqd}, the Continuous Mapping Theorem and Theorem~\ref{lem:dist} imply that for each $\epsilon > 0$,
\begin{align}\label{eq:ll13}
    &\frac{1}{\eta^2}\int_{[-L_\epsilon, L_\epsilon]}\left\{\left(N^{1/3}\left(\hat{\psi}_{N}(r_0 + N^{-1/3}\delta)-\psi_0\right)\right)^{2} \right.\nonumber \\ &\qquad\qquad\qquad\left. -\left(N^{1/3}\left(\hat{\psi}_{N}^{(0)}(r_0 + N^{-1/3}\delta)-\psi_0\right)\right)^{2}\right\}d\delta \nonumber \\ &\stackrel{d}{\to} \frac{1}{\eta^2}\int_{[-L_\epsilon, L_\epsilon]}\left\{\left(g_{\eta, b}(\delta)\right)^{2}-\left(g_{\eta, b}^{0}(\delta)\right)^{2}\right\} d \delta \stackrel{d}{=} \mathbb{D}.
\end{align}
Then, \eqref{eq:ll7}, \eqref{eq:ll9} - \eqref{eq:ll13} and Lemma 4.2 of \cite{MR267677} imply $2 \log \lambda_{N}(h_0) \stackrel{d}{\to} \mathbb{D}$.
\end{proof}

\phantomsection\label{supplementary-material}
\bigskip

\begin{center}
{\large\bf SUPPLEMENTARY MATERIAL}
\end{center}

The R-code for the simulation experiments and real data analysis, described in Section~\ref{sec4} and Section~\ref{sec5} respectively, is available at the following URL: \\ \href{https://github.com/tamojitsadhukhan/Learning-Preference-from-the-Past}{https://github.com/tamojitsadhukhan/Learning-Preference-from-the-Past}.

\end{document}